\documentclass[smallextended]{svjour3}

\usepackage{amsmath,indentfirst,enumitem,amssymb,tikz,color,geometry,float,ragged2e,mathtools,graphicx,url,bigstrut,bm,booktabs,caption}
\usepackage[utf8]{inputenc}
\usepackage{amsfonts}
\newcommand{\R}{\mathbb{R}}

\newcommand{\C}{\mathbf{C}}
\newcommand{\A}{\mathbf{A}}

\newcommand{\GG}{\mathbf{G}}
\newcommand{\HH}{\mathbf{H}}
\newcommand{\MM}{\mathbf{M}}

\newcommand{\M}{\mathbf{M}}

\newcommand{\U}{\mathbf{U}}
\newcommand{\V}{\mathbf{V}}
\newcommand{\Zero}{\mathbf{0}}

\newcommand{\E}{\mathbf{E}}
\newcommand{\F}{\mathbf{F}}
\newcommand{\T}{\mathbf{T}}

\newcommand{\X}{\mathbf{X}}

\newcommand{\W}{\mathbf{W}}
\newcommand{\K}{\mathbf{K}}

\newcommand{\I}{\mathbf{I}}

\date{}

\begin{document}
\author{Silvia Noschese \and Lothar Reichel}
\institute{Silvia Noschese\at
Dipartimento di Matematica\\ 
SAPIENZA Universit\`a di Roma\\
P.le Aldo Moro 5, 00185 Roma, Italy\\
\email{noschese@mat.uniroma1.it} {\it (corresponding author)}\\
Lothar Reichel\at 
Department of Mathematical Sciences\\ 
Kent State University\\
Kent, OH 44242, USA\\
\email{reichel@math.kent.edu}
}
\title{Edge Importance in Complex Networks}
\maketitle

\maketitle

\begin{abstract}
Complex networks are made up of vertices and edges. The latter connect the vertices. There
are several ways to measure the importance of the vertices, e.g., by counting the number 
of edges that start or end at each vertex, or by using the subgraph centrality of the 
vertices. It is more difficult to assess the importance of the edges. One approach is to 
consider the line graph associated with the given network and determine the importance of
the vertices of the line graph, but this is fairly complicated except for small networks. 
This paper compares two approaches to estimate the importance of edges of medium-sized to
large networks. One approach computes partial derivatives of the total communicability of 
the weights of the edges, where a partial derivative of large magnitude indicates that the
corresponding edge may be important. Our second approach computes the Perron sensitivity 
of the edges. A high sensitivity signals that the edge may be important. The performance 
of these methods and some computational aspects are discussed. Applications of interest 
include to determine  whether a network can be replaced by a network with fewer edges with
about the same communicability.
\end{abstract}



\section{Introduction}\label{intro}
Networks are helpful for modeling complex interactions between entities. A network can be 
represented by a graph $\mathcal{G}=\langle\mathcal{V},\mathcal{E},\mathcal{W}\rangle$, 
which consists of a set of \emph{vertices} or \emph{nodes} 
$\mathcal{V}=\{v_1,v_2,\dots,v_n\}$, a set of \emph{edges} 
$\mathcal{E}$, with $|\mathcal{E}|=m$, that connect the vertices, and a set of nonnegative 
edge weights $\mathcal{W}=\{w_{11},w_{21},\ldots,w_{nn}\}$. The weight $w_{ij}$ is 
positive if there is an edge pointing from vertex $v_i$ to vertex $v_j$; while $w_{ij}=0$ 
signifies that there is no such edge. Edges may be directed (and then model one-way 
streets) or undirected (and then model two-way streets). If the graph models a road 
network in which the vertices model intersections and the edges model roads, then the weights
may, e.g., be proportional to the amount of traffic along each road. A network in which 
all positive weights are one is said to be \emph{unweighted}. Descriptions and many 
applications of networks are provided by, e.g., Estrada \cite{Es} and Newman \cite{Ne}. 

We say that vertex $v_i$ is \emph{directly connected} to vertex $v_j$ if there is a single 
edge from vertex $v_i$ pointing to vertex $v_j$. Vertex $v_j$ is then said to be \emph{adjacent}
to vertex $v_i$. When the edge between these vertices is undirected, vertex $v_j$ also is 
directly connected to vertex $v_i$, and $v_i$ is adjacent to vertex $v_j$. Vertex $v_i$ is said 
to be \emph{indirectly connected} to vertex $v_j$ if the latter vertex can be reached from the
former by following at least two edges from $v_i$. We will consider graphs without
multiple edges and without edges that start and end at the same vertex. 

Let $e(v_i\rightarrow v_j)$ denote an edge from vertex $v_i$ to $v_j$. If there also is an 
edge $e(v_j\rightarrow v_i)$ and $w_{ij}=w_{ji}>0$, then the edge is said to be undirected
and denoted by $e(v_i\leftrightarrow v_j)$. A sequence of edges (not necessarily distinct)
\[
\{e(v_1\to v_2),e(v_2\to v_3),\dots,e(v_k\to v_{k+1})\} 
\]
forms a \emph{walk}. The \emph{length} of a walk is the sum of the weights of the edges that make
up the walk, i.e., $\sum_{i=1}^k w_{i,i+1}$. If the edges in a walk are distinct, then the
walk is referred to as a \emph{path}.

Introduce the adjacency matrix $\A=[w_{ij}]_{i,j=1}^n\in\R^{n\times n}$ associated with 
the graph $\mathcal{G}$. The adjacency matrix $\A$ is \emph{sparse} in most applications, 
i.e., the matrix has many more zero entries than positive entries. The matrix is symmetric
if for each edge there also is an edge in the opposite direction with the same weight. The
graph determined by such an adjacency matrix is said to be \emph{undirected}. If at least 
one edge of a graph is directed, or if $w_{ij}\ne w_{ji}$ for at least one index pair 
$\{i,j\}$, then the graph is said to be \emph{directed}. The adjacency matrix associated
with a directed graph is nonsymmetric. Since we assume that there are no edges that start 
and end at the same vertex, the diagonal entries of the adjacency matrix $\A$ vanish. 

A graph is  said to be \emph{connected} if for 
every pair of vertices $v_i$ and $v_j$, there is a path from vertex $v_i$ to vertex $v_j$ and 
from vertex $v_j$ to vertex $v_i$. 
Directed graphs with this property are sometimes referred 
to as \emph{strongly connected}. A directed graph is said to be \emph{weakly connected} if the 
undirected graph that is obtained by replacing all directed edges by undirected ones is 
connected. A graph is strongly connected if and only if the adjacency matrix associated with the 
graph is irreducible. This provides a computational approach to determine whether a graph is 
strongly connected. Further, an undirected graph is connected if and only if  the second smallest 
eigenvalue of the associated graph Laplacian is positive; see e.g., \cite{Es,Ne}.

Let $\exp_0(t)=\exp(t)-1$ and consider the power series expansion
\begin{equation}\label{expan}
\exp_0(\A)=\A+\frac{\A^2}{2!}+\frac{\A^3}{3!}+\ldots~.
\end{equation}
Let $\A^k=[a_{ij}^{(k)}]_{i,j=1}^n$ for $k=1,2,\ldots~$, where $a_{ij}^{(1)}=w_{ij}$ for 
$1\leq i,j\leq n$. A nonvanishing entry $a_{ij}^{(k)}$ for some $k>0$ indicates that there
is at least one walk of $k$ edges from vertex $v_i$ to vertex $v_j$. In case of an
unweighted network, $a_{ij}^{(k)}$ represents the number of the walks of length $k$ from 
 vertex $v_i$ to vertex $v_j$. The denominators in the  terms of the expansion 
 \eqref{expan} ensure that the expansion converges and that terms 
$\frac{\A^k}{k!}$ with $k$ large contribute only little to $\exp_0(\A)$. It follows that 
short walks typically are more important than long ones, which is in agreement with the 
intuition that messages propagate better along short walks than along long ones. This led
Estrada and Rodriguez-Velazquez \cite{ER} to use $\exp(\A)$ to study properties of a 
graph; we use the function $\exp_0(\A)$ because the term with the identity matrix 
$\mathbf{I}$ in the expansion of $\exp(\A)$ has no natural interpretation in the context 
of network modeling. Other functions also can be used such as a resolvent or a 
Mittag-Leffler function; see Estrada and Higham \cite{EH} and Arrigo and Durastante 
\cite{AD} for discussions. 

Estrada and Rodriguez-Velazquez \cite{ER} define for graphs with a symmetric adjacency 
matrix the communicability matrix $\exp(\A)$; we will use the matrix
\[
\C=[c_{ij}]_{i,j=1}^n=\exp_0(\A).
\]
The entry $c_{ij}$, $i\ne j$, is referred to as the \emph{communicability} between the 
vertices $v_i$ and $v_j$. A relatively large value implies that it is easy for the vertices 
$v_i$ and $v_j$ to communicate. Estrada and Rodriguez-Velazquez \cite{ER} measure the 
importance of the vertex $v_i$ of an undirected graph by the \emph{subgraph centrality} 
$c_{ii}+1$; we will use $c_{ii}$. Related measures of communicability can be defined when 
the adjacency matrix $\A$ is nonsymmetric; see \cite{DLCMR}.

We define the \emph{total communicability} of the graph $\mathcal{G}$ as
\begin{equation}\label{TC}
T_{\mathcal G}(w_{11},w_{21},\ldots,w_{nn})=e^T\exp_0(\A)e,
\end{equation}
where $e=[1,1,\ldots,1]^T\in\R^n$ denotes the vector with all entries $1$ and the 
superscript $^T$ stands for transposition. Benzi and Klymko \cite{BK} introduced the
related measure $e^T\exp(\A)e$, which differs from \eqref{TC} by the additive constant 
$n$. Note that the expression \eqref{TC} is invariant under transposition. We have
\[
T_{\mathcal G}(w_{11},w_{21},\ldots,w_{nn})=e^T\exp_0(\A^T)e=
\frac{1}{2}e^T(\exp_0(\A)+\exp_0(\A^T))e.
\]
The graph associated with the adjacency matrix $\A^T$ is known as the \emph{reverse graph}
to ${\mathcal G}$. Thus, the total communicability of the graph ${\mathcal G}$ and of the
reverse graph are the same.

We are interested in investigating the importance of the edges of a graph ${\mathcal G}$
and, in particular, in determining which edge weights can be reduced or set to zero 
without significantly affecting the total communicability. A possible approach to 
investigate the importance of edges is to consider the line graph associated with the 
graph. The edges of ${\mathcal G}$ correspond to vertices in the associated line graph, and
the importance of the edges in ${\mathcal G}$ can be measured by the subgraph centrality 
of the vertices of the line graph. This approach is investigated in \cite{DLCMR2,DLCMR3}. 
However, it is quite cumbersome to construct the line graph except for small graphs. A 
simple heuristic technique was proposed by Arrigo and Benzi \cite{AB}, who define the 
edge total communicability and seek to remove edges whose removal does not reduce the
edge total communicability much. This approach is quite easy to implement for networks 
that are small enough to allow the computation of the singular value decomposition of the
associated adjacency matrix; however, it is not straightforward to use for large-scale 
networks. Moreover, the importance of an edge is assumed to depend on the importance of
the vertices that it connects. This holds for some networks, but not for others. Therefore,
this approach to edge removal may result in removals that are not in agreement with 
intuition. More recently an approach that uses
the right and left Perron vectors of the adjacency matrix in combination with the Wilkinson 
perturbation to determine which weights to increase in order to increase the total 
communicability has been described in \cite{DLCJNR,NR}. An analogous technique is applied 
in \cite{EHNR} to discern which weights in a weighted multilayer network can be decreased 
without affecting the total communicability significantly. 

Another approach to study the importance of edges is to evaluate the Fr\'echet derivatives
of the total communicability \eqref{TC} with respect to the weights. This approach was 
advocated by De la Cruz Cabrera et al. \cite{DLCJNR} and recently Schweitzer \cite{Sc}
described how to speed up the computations. Introduce the gradient
\begin{equation}\label{grad}
\nabla T_{\mathcal G}(w_{11},w_{21},\ldots,w_{nn})=
\left[\frac{\partial}{\partial w_{11}}e^T\exp_0(\A)e,
\frac{\partial}{\partial w_{21}}e^T\exp_0(\A)e,\ldots,
\frac{\partial}{\partial w_{nn}}e^T\exp_0(\A)e\right]^T.
\end{equation}
When the partial derivative 
\begin{equation}\label{partderive}
\frac{\partial}{\partial w_{ij}}e^T\exp_0(\A)e 
\end{equation}
is (relatively) large, a small increase in the positive weight $w_{ij}$ results in a 
substantial change in the total communicability. We will show below that the partial 
derivatives \eqref{partderive} are nonnegative. 

For example, let the graph represent a 
road map, where the edges model roads and the vertices represent intersections of roads. 
Let vertex $v_j$ be adjacent to vertex $v_i$. Then widening an existing road from $v_i$ to
$v_j$ may result in a significant increase in the total communicability of the graph; the 
widening of this road is modeled by increasing the weight $w_{ij}$. Also, when 
\eqref{partderive} is relatively large, and vertex $v_j$ is not adjacent to vertex $v_i$, 
i.e., there is no road from $v_i$ to $v_j$, building such a road may increase the total 
communicability substantially. This is modeled by making the vanishing weight $w_{ij}$
positive. 

Conversely, if the partial derivative \eqref{partderive} is relatively small and the 
weight $w_{ij}>0$ is small, then setting $w_{ij}$ to zero, i.e., removing the edge from
vertex $v_i$ to vertex $v_j$, will not affect the total communicability \eqref{TC} much. 
This implies that blocking the road from vertex $v_i$ to $v_j$, e.g., due to construction,
does not change the total communicability of the network significantly. 

This paper is organized as follows. Section \ref{sec2} discusses how the gradient can be
applied to assess the importance of edges. The first part of the section is concerned with
small to medium-sized problems for which it is feasible to evaluate the gradient
\eqref{grad}. The latter part of the section discusses the application of Krylov subspace
methods to project large-scale problems to problems of fairly small size. The computations
use a result by Schweitzer \cite{Sc} on the evaluation of Fr\'echet derivatives, but 
differ in various aspects that speed up the computations. Section \ref{sec3} reviews 
methods described in 
\cite{DLCJNR,EHNR,NR} based on evaluating the right and left Perron vectors and the 
Wilkinson perturbation to determine important and unimportant edges. One of the
aims of this paper is to compare the methods discussed in Sections \ref{sec2} and
\ref{sec3}. This is done in Section \ref{sec4}, where we also report timings.
Section \ref{sec6} contains concluding remarks.

We conclude this section with comments on some related methods. A scheme that combines 
regression, soft-thresholding, and projection is applied in \cite{DLCJR} to approximate an
unweighted network by a simpler unweighted network. This scheme performs well but may be 
expensive and is restricted to unweighted networks. Massei and Tudisco \cite{MT} consider 
the problem of determining a low-rank perturbation $\E\in\R^{n\times n}$ to the adjacency 
matrix $\A$ so that the perturbed matrix $\A+\E$ maximizes or minimizes the robustness 
of the network.
For instance, $\E$ may be chosen to maximize or minimize the trace of $f(\A+\E)-f(\A)$
for a user-specified matrix-valued function $f$. Thus, this method seeks to
modify a few edge weights so that the trace is increased or decreased as much as possible.
The perturbation $\E$ is determined by a greedy algorithm for solving an optimization 
problem that can be quite expensive to solve. A careful comparison with this method is
outside the scope of the present paper.

\section{Network modifications based on the gradient}\label{sec2} 
This section discusses methods for modifying, adding, or removing edges of a network by
using information furnished by the gradient \eqref{grad}. We first describe methods for 
small to medium-sized networks for which all entries of the gradient \eqref{grad} can be 
evaluated. Subsequently, we will consider Krylov subspace methods that can be applied to 
large-scale networks.

\subsection{Methods for small to medium-sized networks}\label{subsec2.1}
Let the matrix function $f:\A\in\R^{n\times n}\to f(\A)\in\R^{n\times n}$ be continuously 
differentiable sufficiently many times in a region in the complex plane that contains all
eigenvalues of $\A$. Then the function $f$ has a Fr\'echet derivative 
$L_f(\A,\E)\in\R^{n\times n}$ at $\A$ in the direction 
 $\E\in\R^{n\times n}\backslash\{\Zero\}$.
The Fr\'echet derivative satisfies
\begin{equation}\label{frechet}
f(\A+\E)=f(\A)+L_f(\A,\E)+o(\|\E\|),
\end{equation}
as $\|\E\|\to 0$, where $\|\cdot\|$ is any matrix norm; see, e.g., \cite{Hi} for details. 
Schweitzer described an efficient approach to evaluate $L_f(\A,\E)$ in several directions 
$\E$ simultaneously.

\begin{theorem}(Schweitzer \cite[Theorem 2.3]{Sc})
Let $\A\in\R^{n\times n}$ and $u,v\in\R^n\backslash\{0\}$, and assume that $f$ is 
Fr\'echet differentiable at $\A$. Define $\E_{ij}=e_ie_j^T$, where 
$e_k=[0,\ldots,0,1,0,\ldots,0]^T\in\R^n$ denotes the $k$th column of the identity matrix.
Then
\begin{equation}\label{rhs}
u^TL_f(\A,\E_{ij})v=e_i^TL_f(\A^T,uv^T)e_j.
\end{equation}
\end{theorem}

Thus, the entries of the matrix $L_f(\A^T,uv^T)$ furnish Fr\'echet derivatives in all 
directions $\E_{ij}=e_ie_j^T$, $1\leq i,j\leq n$. We are primarily interested in the 
situation when $f(t)=\exp_0(t)$ and 
\begin{equation}\label{uv}
u=v=e=[1,1,\ldots,1]^T\in\R^n. 
\end{equation}

\begin{theorem}
All entries of the gradient \eqref{grad} are nonnegative.
\end{theorem}

\begin{proof}
Let $\E_{ij}=e_ie_j^T$. Then for $f(t)=\exp_0(t)$, we have
\[
\frac{\partial}{\partial w_{ij}}e^T\exp_0(\A)e=e^T L_f(\A,\E_{ij})e.
\]
It follows from \eqref{frechet} that
\[
L_f(\A,h\E_{ij})=\exp_0(\A+h\E_{ij})-\exp_0(\A)+o(h)~~\mbox{as}~~h\searrow 0.
\]
The power series expansion of $f(t)=\exp_0(t)$ gives
\begin{eqnarray*}
L_f(\A,h\E_{ij})&=&h\left(\E_{ij}+\frac{\A\E_{ij}+\A\E_{ij}}{2!}+
\frac{\A^2\E_{ij}+\A\E_{ij}\A+\E_{ij}\A^2}{3!}+\ldots~\right)+o(h)\\
&=&h\left(\sum_{\ell=1}^\infty \sum_{k=0}^{\ell-1}\A^k\E_{ij}\A^{\ell-1-k}+o(1)\right).
\end{eqnarray*}
Each term in the above sum is a matrix with nonnegative entries. Hence, the sum is a
matrix with nonnegative entries. The term $o(1)$ vanishes as $h\searrow 0$. Since
$L_f(\A,h\E_{ij})$ is linear in $h$, we obtain
\begin{equation}\label{Lfexpan}
L_f(\A,\E_{ij})=\lim_{h\searrow 0}\frac{L_f(\A,h\E_{ij})}{h}=
\sum_{\ell=1}^\infty \sum_{k=0}^{\ell-1}\A^k\E_{ij}\A^{\ell-1-k}.
\end{equation}
This completes the proof.
\end{proof}

A possible way to evaluate the matrix $L_f(\A^T,uv^T)$ in \eqref{rhs} when 
$\A^T\in\R^{n\times n}$ is to use the relation
\begin{equation}\label{Lf}
f\left(\left[\begin{array}{cc} \A^T  & uv^T \\ \mathbf{0} & \A^T \end{array} 
\right]\right)=\left[\begin{array}{cc} f(\A^T)  & L_f(\A^T,uv^T) \\ \mathbf{0} & 
f(\A^T) \end{array} \right];
\end{equation}
see, e.g., \cite[p.\,253]{Hi}. However, when $f(t)=\exp_0(t)$, the computation of 
$f(\A^T)$ requires ${\mathcal O}(n^3)$ arithmetic floating point operations (flops). 
Therefore, the evaluation of the left-hand side of \eqref{Lf} demands about $8$ times more
flops than the calculation of $f(\A^T)$. It is cheaper to approximate $L_f(\A^T,uv^T)$ by
using the finite-difference approximation
\begin{equation}\label{Lfapprox}
L_f(\A^T,uv^T)\approx\frac{f(\A^T+huv^T)-f(\A^T-huv^T)}{2h}
\end{equation}
for some $h>0$. We will use $h=\frac{2}{n}\cdot 10^{-4}$ in the computed examples in
Section \ref{sec4}. This is suggested by the following simple computations. We have used 
the fact that $\|uv^T\|_2=n$, which holds for the vectors \eqref{uv}. Here and throughout 
this paper $\|\cdot\|_2$ denotes the spectral matrix norm or the Euclidean vector norm.

Example 2.1. Let $f(t)=\exp_0(t)$ be evaluated with a relative error $\delta_t$ bounded 
by $\delta>0$ and let $h>0$ be a small scalar. Then 
\begin{eqnarray*}
\frac{f(t+h)-f(t-h)}{2h}&\approx&
\frac{f_{\text{exact}}(t+h)+\delta_{t+h}f_{\text{exact}}(t)-(f_{\text{exact}}(t-h)
+\delta_{t-h}f_{\text{exact}}(t))}{2h}\\
&\approx& f_{\text{exact}}'(t)+\frac{h^2}{6}f_{\text{exact}}'''(t)+
\frac{\delta_{t+h}-\delta_{t-h}}{2h}f_{\text{exact}}(t).
\end{eqnarray*}
Thus, the error is bounded by about
\[
\left(\frac{h^2}{6}+\frac{\delta}{h}\right)\exp(t).
\]
Minimization over $h>0$ yields
\[
h\approx (3\delta)^{1/3}.
\]
The computation of the scalar exponential is carried out with high relative accuracy in 
MATLAB. However, evaluation of the matrix exponential $\exp(\A^T)$ is more difficult. It 
can be computed in several ways; see, e.g., \cite[Chapter 10]{Hi} as well as 
\cite{RSID,XX}.  The accuracy achieved depends on the method used as well as on the size 
and properties of the matrix $\A^T$; see, e.g., \cite[Chapter 10]{Hi} and \cite{RSID} for
computed examples. We therefore include a factor  $10^3$ in the bound $\delta$ for the 
relative accuracy. This bound is valid for most matrices of sizes of interest to us. 
Letting $\delta\approx 10^3\epsilon_{\text{mach}}$ with 
$\epsilon_{\text{mach}}\approx 2\cdot 10^{-16}$ gives $h\approx 2\cdot 10^{-4}$.~~~$\Box$

\begin{proposition}
Let $f(t)=\exp_0(t)$. Then 
\[
\frac{f(\A^T+h\E_{ij})-f(\A^T-h\E_{ij})}{2h}=L_f(\A^T,\E_{ij})+O(h^2).
\]
\end{proposition}

\begin{proof}
The right-hand side of \eqref{Lfapprox} with $uv^T$ replaced with $\E_{ij}$ can be 
expressed as
\begin{eqnarray}
\nonumber
\frac{f(\A^T+h\E_{ij})-f(\A^T-h\E_{ij})}{2h}&=&\E_{ij}+\frac{1}{2!}(\A^T\E_{ij}+\E_{ij}\A^T)\\
\nonumber
&+& \frac{1}{3!}((\A^T)^2\E_{ij}+\A^T\E_{ij}\A^T+\E_{ij}(\A^T)^2)\\
\nonumber
&+&\frac{1}{4!}((\A^T)^3\E_{ij}+(\A^T)^2\E_{ij}\A^T+\A^T\E_{ij}(\A^T)^2+\E_{ij}(\A^T)^3)\\
\label{lastterm}
&+& \ldots~+O(h^2).
\end{eqnarray}
The result now follows from \eqref{Lfexpan} with $\A$ replaced by $\A^T$.
\end{proof}

The above proposition is a corollary of the well-known series representation of 
the Fr\'echet derivative; see, e.g., \cite{Hi}. It can be stated for an arbitrary direction 
matrix $\E$ by just replacing the last term $O(h^2)$ in \eqref{lastterm} with
$O(\|\E\|_2\,h^2)$.  The evaluation of the right-hand side of 
\eqref{Lfapprox} with $f(t)=\exp_0(t)$ gives
approximations of all the entries of the gradient \eqref{grad}. We will refer to 
$\|\nabla T_{\mathcal G}(w_{11},w_{21},\ldots,w_{nn})\|_2$ as the 
\emph{total transmission} of the graph ${\mathcal G}$.

\subsubsection{Network simplification by edge removal}\label{subsubsec211}
One of the aims of this paper is to discuss how to reduce the complexity of a network by 
removing edges without changing the total transmission or total communicability 
significantly. A simple way to achieve the former is to set positive weights $w_{ij}$ to 
zero when the associated entries of the gradient \eqref{partderive} are (relatively) 
small, thus removing the corresponding edges $e(v_i\to v_j)$. This determines a new 
network $\widetilde{\mathcal G}$ with fewer edges than ${\mathcal G}$ with about the same
total transmission. However, in order for the network $\widetilde{\mathcal G}$ also to 
have about the same total communicability as ${\mathcal G}$, we also have to require that 
the removed weights be small. We therefore introduce the vector 
${\mathcal{E}}_{L_{f}}\in\R^m$, whose $k$th entry is the {\it edge importance} of 
$e_k=e(v_i\to v_j)$, defined as the product of the weight $w_{ij}$ and the corresponding 
partial derivative \eqref{partderive} normalized by the total transmission. Observe that
$||L_f(\A^T,ee^T)||_F=||\nabla T_{\mathcal G}||_2$, where $\|\cdot\|_F$ stands for the 
Frobenius norm. We refer to the norm $\|{\mathcal{E}}_{L_{f}}\|_2$ as the 
\emph{total edge importance}. 

The following simple procedure can be used to construct the
edge importance vector ${\mathcal{E}}_{L_{f}}$ for undirected graphs:\\

\noindent
{\bf Procedure 1}:
\begin{enumerate}
\item Multiply the adjacency matrix $\A$ element by element by the  matrix 
$L_f(\A^T,ee^T)$. \\
\item Divide the elements of the so obtained matrix that correspond to edges of 
${\mathcal G}$ by the total transmission $||\nabla T_{\mathcal G}||_2$ and put them column
by column into the vector ${\mathcal{E}}_{L_{f}}$.
\end{enumerate}

If the graph is undirected, then the $k$th entry of the vector 
${\mathcal{E}}_{L_{f}}\in\R^m$ gives the importance of the edge 
$e_k=e(v_i\leftrightarrow v_j)$. We obtain the following procedure:\\

\noindent
{\bf Procedure 2}:
\begin{enumerate}
\item Extract the strictly lower triangular portion $\mathbf{L}$ of the adjacency matrix 
$\A$, and multiply $\mathbf{L}$ element by element by the strictly lower triangular portion 
of the matrix $L_f(\A^T,ee^T)$. 
\item Divide the elements of the matrix so obtained that correspond to the edges 
$e(v_i\to v_j)$ of ${\mathcal G}$ with $i>j$ by $\frac{||\nabla T_{\mathcal G}||_2}{2}$
and put them column by column into the vector ${\mathcal{E}}_{L_{f}}$.
\end{enumerate}

Consider the cone $\mathcal{A}$ of all nonnegative matrices in $\R^{n\times n}$ with the 
same sparsity structure as $\A$ and let $\M|_{\mathcal{A}}$ denote a matrix in 
$\mathcal{A}$ that is closest to a given nonnegative matrix $\M\in\R^{n\times n}$ with respect to the 
Frobenius norm. It is straightforward to verify that $\M|_{\mathcal{A}}$ is obtained by 
setting all the entries outside the sparsity structure of $\A$ to zero.

In the first step of the Procedure 1, one considers the matrix 
$L_f(\A^T,ee^T)|_{\mathcal{A}}$, whereas in the first step of the second procedure
one considers the projected matrix $L_f(\A^T,ee^T)|_{\mathcal{L}}$ with ${\mathcal{L}}$ 
the cone of all nonnegative matrices in $\R^{n\times n}$ with the same sparsity structure 
as $\mathbf{L}$. It follows that $\mathbf{L}=\A|_{\mathcal{L}}$.

In computations, we order the entries of ${\mathcal{E}}_{L_{f}}$ from smallest to largest 
and set the weights $w_{ij}$ (or $w_{ij}=w_{ji}$ if the graph is undirected) associated 
with the first few of the ordered edge importances to zero. Some post-processing may be 
necessary if the reduced graph $\widetilde{\mathcal G}$ obtained by removing the edges
associated with the weights that are set to zero is required to be connected. 

\subsubsection{Network modification to increase or decrease total communicability}
\label{subsubsec212}
We turn to the task of increasing or decreasing the total communicability of a network by
changing a few weights. The weights to be changed are chosen with the aid of the entries 
of the vector ${\mathcal{E}}_{L_{f}} \in \R^m$. We obtain a relatively large 
increase/reduction in the total communicability by slightly increasing/reducing the 
weights associated with the largest entries of ${\mathcal{E}}_{L_{f}}$. To this end, we 
order the entries of ${\mathcal{E}}_{L_{f}}$ from the largest to the smallest. More than 
one of the weights can be modified to achieve a desired increase or reduction in the total
communicability.

Assume that the given graph ${\mathcal G}$ is strongly or weakly connected, and that we
would like the modified graph to have the same property. Consider the situation when 
removing an edge $e_k$ associated with one of the first few of the ordered entries of the 
vector ${\mathcal{E}}_{L_{f}}$ results in a graph that does not have this property. Then
typically the total communicability can be decreased considerably by reducing the 
corresponding edge-weight $w_{ij}$ to a small positive value (or both the weights $w_{ij}$
and $w_{ji}$ to the same small positive value if the graph is undirected), with the 
perturbed graph so obtained having the same connectivity property as the original graph.

\subsubsection{Network modification by inclusion of new edges}\label{subsec5}
The partial derivatives \eqref{partderive} reveal which edges would be important to add to
a given graph to increase the communicability of the network significantly, namely 
nonexistent edges, whose associated partial derivative is large. Let 
$\widehat{\mathcal{A}}$ denote the cone of the nonnegative matrices in $\R^{n\times n}$ 
with sparsity structure given by the zero entries of $\A$ except for the diagonal entries.
The {\it virtual importance} of the nonexistent edge $e(v_i\to v_j)\notin {\mathcal E}$ is
given by the corresponding entry of $L_f(\A^T,ee^T)$ normalized by the total transmission. 
The construction of the virtual edge importance vector 
$\widehat{{\mathcal{E}}}_{L_{f}}\in\R^{n^2-n-m}$, which makes use of matrix entries
in the sparsity structure associated with $\widehat{\mathcal{A}}$, can be summarized 
as follows:\\

\noindent
{\bf Procedure 3}:
\begin{enumerate}
\item Construct the matrix $L_f(\A^T,ee^T)|_{\widehat{\mathcal{A}}}$.
\item Divide the entries of the matrix so obtained that belong to the sparsity structure 
associated with $\widehat{\mathcal{A}}$ by the total transmission  
$||\nabla T_{\mathcal G}||_2$ and put them column by column into the vector
$\widehat{{\mathcal{E}}}_{L_{f}}$.
\end{enumerate}

If the graph is undirected, then the virtual importance of the virtual edge 
$e(v_i\leftrightarrow v_j)\notin {\mathcal E}$ is defined as twice the corresponding 
entry in $L_f(\A^T,ee^T)$ normalized by the total transmission. Let 
$\widehat{\mathcal{L}}$ be the cone of the nonnegative matrices in 
$\widehat{\mathcal{A}}$, where all the entries in the strictly upper triangular portion
are set to zero. The procedure for the construction of the virtual edge importance vector 
$\widehat{{\mathcal{E}}}_{L_{f}}\in \R^{(n^2-n-2m)/2}$, which makes use of matrix entries
in the sparsity structure associated with $\widehat{\mathcal{L}}$, becomes:\\

\newpage
\noindent
{\bf Procedure 4}:
\begin{enumerate}
\item Construct the matrix $L_f(\A^T,ee^T)|_{\widehat{\mathcal{L}}}$.
\item Divide the entries of the matrix so obtained that belong to the sparsity structure 
associated with $\widehat{\mathcal{L}}$ by $\frac{||\nabla T_{\mathcal G}||_2}{2}$ and put
them column by column into the vector $\widehat{{\mathcal{E}}}_{L_{f}}$.
\end{enumerate}

In case $\|L_f(\A^T,ee^T)|_{\mathcal{A}}\|_F \ll ||\nabla T_{\mathcal G}||_2$ many 
 large partial derivatives \eqref{partderive} correspond to zero extra-diagonal entries 
of $\A$. Then making suitable zero weights 
$w_{ij}$ positive, or if the graph is undirected giving suitable pairs of zero weights 
$\{w_{ij},w_{ji}\}$ the same positive value, might be beneficial. We recall that 
replacing a zero weight $w_{ij}$ with a positive weight $w_{ij}$ is equivalent 
to including a weighted edge $e(v_i\leftrightarrow v_j)$ into the graph.

In computations, we order the entries of $\widehat{{\mathcal{E}}}_{L_{f}}$ from largest to 
smallest and add the positive weights $w_{ij}$ (or pairs of positive entries $w_{ij}=w_{ji}$ 
if the graph is undirected) associated with the first few of the virtual ordered edge 
importances.

\subsection{Methods for large networks}\label{subsec2.2}
Recently, Kandolf et al. \cite{KKRS} derived a method for evaluating approximations of the
Fr\'echet derivative of a matrix function by Krylov subspace methods. Applications of this
technique to network analysis have recently been discussed by De la Cruz Cabrera et al. 
\cite{DLCJNR} and Schweitzer \cite{Sc}. We first outline this method and subsequently
discuss some alternatives. 

Let $\A\in\R^{n\times n}$ and assume that the function $f$ is analytic in an open simply
connected set $\Omega$ in the complex plane that contains the spectrum of $\A$. Then
\[
f(\A)=\frac{1}{2\pi\mathrm{i}}\int_\Gamma f(z)(z\I-\A)^{-1}dz,
\]
where $\Gamma$ is a curve in $\Omega$ that winds around the spectrum of $\A$ exactly once
and $\mathrm{i}=\sqrt{-1}$. In this paper, we are primarily interested in the situation 
when $f(z)=\exp_0(z)$, but the techniques discussed apply to other analytic functions as 
well. Let $u,v\in\R^n$ be nonvanishing vectors. The 
Fr\'echet derivative of $f$ at $\A$ in the direction $uv^T$ can be expressed as
\[
L_f(\A,uv^T)=\frac{1}{2\pi\mathrm{i}}\int_\Gamma f(z)(z\I-\A)^{-1}uv^T(z\I-\A)^{-1}dz;
\]
see, e.g., \cite{Hi}. Kandolf et al. \cite{KKRS} determine an approximation of this expression by using Krylov subspace techniques to
approximate the vectors
\[
s(z)=(z\I-\A)^{-1}u,\qquad t(z)=(z\I-\A)^{-H}v,\qquad z\in\Gamma,
\]
where the superscript $^H$ denotes transposition and complex conjugation. Specifically,
Kandolf et al.
\cite{KKRS} and Schweitzer \cite{Sc} approximate the vectors $s(z)$ and $t(z)$ by a Krylov
subspace technique based on the Arnoldi process. Application of $1\leq \ell\ll n$ steps of 
the Arnoldi process to $\A$ with initial vector $u$, and to $\A^T$ with initial vector 
$v$, generically, yields the Arnoldi decompositions
\begin{equation}\label{arndec}
\A\V_{\ell}=\V_{\ell}\GG_{\ell}+{\tilde v}_{{\ell}+1}e_{\ell}^T,\qquad
\A^T\W_{\ell}=\W_{\ell}\MM_{\ell}+{\tilde w}_{{\ell}+1}e_{\ell}^T,
\end{equation}
where the matrices $\GG_{\ell},\MM_{\ell}\in\R^{{\ell}\times {\ell}}$ are of upper Hessenberg form, the matrix 
$\V_{\ell}\in\R^{n\times {\ell}}$ has orthonormal columns with initial column $u/\|u\|_2$, the 
vector ${\tilde v}_{{\ell}+1}\in\R^n$ is orthogonal to the range of $\V_{\ell}$, the matrix 
$\W_{\ell}\in\R^{n\times{\ell}}$ has orthonormal columns with initial column $v/\|v\|_2$, and the 
vector ${\tilde w}_{{\ell}+1}\in\R^n$ is orthogonal to the range of $\W_{\ell}$; see, e.g., Saad 
\cite[Chapter 6]{Sa} for details on the Arnoldi process. Here we only note that the 
computation of the decompositions \eqref{arndec} requires the evaluation of ${\ell}$ 
matrix-vector products with the matrix $\A$ and ${\ell}$ matrix-vector products with the matrix
$\A^T$. This is the dominating computational work when the matrix $\A$ is large and the 
number of Arnoldi steps is fairly small. We assume here that the Arnoldi processes do not 
break down when computing \eqref{arndec}; in case of breakdown, the formulas 
\eqref{arndec} simplify. 

Kandolf et al. \cite{KKRS} propose to use the Arnoldi approximation
\begin{equation}\label{Lfapprox2}
L_{f,{\rm Arnoldi}}(\A,uv^T)=\V_{\ell}\X_{\ell}\W_{\ell}^T,
\end{equation}
of $L_f(\A,uv^T)$, where $\X_{\ell}$ is the upper right ${\ell}\times {\ell}$ submatrix of the 
$2{\ell}\times 2{\ell}$ matrix
\begin{equation}\label{Lfsmall}
f\left(\left[\begin{array}{cc} \GG_{\ell}  & \;\|u\|_2\|v\|_2\,e_1e_1^T \\ \mathbf{0} & \MM_{\ell}^T
\end{array} \right]\right)=\left[\begin{array}{cc} f(\GG_{\ell})  & \X_{\ell} \\ \mathbf{0} & 
f(\MM_{\ell}^T) \end{array} \right]
\end{equation}
with $e_1=[1,0,\ldots,0]^T$. Schweitzer \cite{Sc} applies formulas \eqref{arndec}, 
\eqref{Lfapprox2}, and \eqref{Lfsmall} with $\A$ replaced by $\A^T$ and $u=v=e$; cf. 
\eqref{rhs}. Convergence results are provided by Kandolf et al. \cite{KKRS}. 

An alternative approach to compute an approximation of the Fr\'echet derivative 
$L_f(\A^T,ee^T)$ is to apply finite-difference approximations analogously as
\eqref{Lfapprox}. Application of ${\ell}$ steps of the Arnoldi process to the matrix $\A^T$ 
with initial vector $e$ gives the Arnoldi decomposition
\[
\A^T \U_{\ell}=\U_{\ell} \HH_{\ell}+{\tilde u}_{{\ell}+1}e_{\ell}^T,
\]
where the columns of $\U_{\ell}\in\R^{n\times {\ell}}$ are orthonormal and span the Krylov subspace
\[
\K_{{\ell}}(\A^T,e)={\rm span}\{e,\A^Te,(\A^T)^2e,\ldots,(\A^T)^{{\ell}-1}e\}
\]
and the first column of $\U_{\ell}$ is $e/\sqrt{n}$. Moreover, the vector ${\tilde u}_{{\ell}+1}\in\R^n$ is 
orthogonal to $\K_{\ell}(\A^T,e)$ and $\HH_{\ell}\in\R^{{\ell}\times {\ell}}$ is an upper Hessenberg matrix.
Then
\[
\U_{\ell}^T\A^T\U_{\ell}=\HH_{\ell}.
\]
We will use the approximations
\begin{equation}\label{ATapprox}
\A^T\approx \U_{\ell}\HH_{\ell}\U_{\ell}^T 
\end{equation}
and
\begin{equation}\label{fA}
f(\A^T)\approx \U_{\ell}f(\HH_{\ell})\U_{\ell}^T.
\end{equation}
The approximation \eqref{fA} is quite accurate when $f(\A^T)$ can be approximated well by 
a matrix of low rank. This is the case for many real-life undirected networks when 
$f(t)=\exp(t)$; see \cite{FMRR}. 

It follows from \eqref{ATapprox} that
\[
\A^T+hee^T\approx \U_{\ell}\HH_{\ell}\U_{\ell}^T + hee^T = \U_{\ell}(\HH_{\ell} + hne_1e_1^T)\U_{\ell}^T
\]
and from \eqref{fA} that
\[
f(\A^T+hee^T)\approx \U_{\ell} f(\HH_{\ell} + hne_1e_1^T)\U_{\ell}^T,
\]
where $h$ is a scalar of small magnitude  and $e_1=[1,0,\ldots,0]^T\in\R^m$. Hence,
\begin{equation}\label{Lffd}
\frac{f(\A^T+hee^T)-f(\A^T-hee^T)}{2h}\approx \U_{\ell} 
\frac{f(\HH_{\ell} + hne_1e_1^T)-f(\HH_{\ell} - hne_1e_1^T)}{2h}\U_{\ell}^T.
\end{equation}
We will use the expression on the right-hand side as an approximation of $L_f(\A^T,ee^T)$
in computed examples with $h$ the same as in Example 2.1. Note that the evaluation of this 
expression only requires the computation of $\ell$ matrix-vector products with the matrix 
$\A^T$. For large-scale problems for which the evaluation of matrix-vector products is the 
dominant computational work, under the assumption that both methods require the same 
number of iterations, the use of the right-hand side of \eqref{Lffd} halves the
computational burden when compared with the evaluation of \eqref{Lfapprox2}.

We turn to the situation when the matrix $\A\in\R^{n\times n}$ is symmetric and first
review the computations described by Kandolf et al. \cite{KKRS} of the analogue of the 
expression \eqref{Lfapprox2} when the direction is $ee^T$. Then the calculation of the 
Arnoldi decompositions \eqref{arndec} can be replaced by application of ${\ell}$ steps of the
symmetric Lanczos process to $\A$ with initial vector $e$. Generically, we obtain 
\[
\A\U_{\ell}=\U_{\ell}\T_{\ell}+{\tilde u}_{{\ell}+1}e_{\ell}^T,
\]
where the matrix $\T_{\ell}\in\R^{{\ell}\times {\ell}}$ is symmetric and tridiagonal, the matrix 
$\U_{\ell}\in\R^{n\times {\ell}}$ has orthonormal columns with initial column $e/\|e\|_2$, and the
vector ${\tilde v}_{{\ell}+1}\in\R^n$ is orthogonal to the range of $\U_{\ell}$; see, e.g., Saad
\cite{Sa0} for details on the symmetric Lanczos process. 

The analogue of the expression \eqref{Lfapprox2} is given by
\begin{equation}\label{Lfapprox2L}
L_{f,{\rm Lanczos}}(\A,ee^T)=\U_{\ell}\X_{\ell}\U_{\ell}^T,
\end{equation}
where $\X_{\ell}$ is the upper right ${\ell}\times {\ell}$ submatrix of the $2{\ell}\times 2{\ell}$ matrix
\begin{equation}\label{Lfapproxsmall}
f\left(\left[\begin{array}{cc} \T_{\ell}  & ne_1e_1^T \\ \mathbf{0} & \T_{\ell}
\end{array} \right]\right)=\left[\begin{array}{cc} f(\T_{\ell})  & \X_{\ell} \\ \mathbf{0} & 
f(\T_{\ell}) \end{array} \right];
\end{equation}
see \cite{KKRS} for further details. 

It remains to discuss how to determine approximations of the elements of $L_f(\A^T,ee^T)$
of largest and smallest magnitude by using the right-hand sides of \eqref{Lfapprox2} or
\eqref{Lffd}. We first consider the former. To determine an approximation of an entry of
largest magnitude of $L_f(\A^T,ee^T)$, we first locate an entry $x_{ij}$ of the matrix 
$\X_{\ell}$ of largest magnitude and then determine entries of largest magnitude of columns
$i$ and $j$ of the matrices $\V_{\ell}$ and $\W_{\ell}$, respectively. The product of these entries 
furnishes an approximation of an entry of $L_f(\A^T,ee^T)$ of largest magnitude. We 
proceed analogously to determine an approximation of an entry of $L_f(\A^T,ee^T)$ of 
smallest magnitude. Other entries of closest to largest or smallest magnitudes can be 
computed similarly. 

We turn to the use of the right-hand side of \eqref{Lffd}. To determine an approximation 
of an entry of largest magnitude of $L_f(\A^T,ee^T)$, we first locate an entry of the 
matrix $\frac{f(\HH_{\ell} + hne_1e_1^T)-f(\HH_{\ell} - hne_1e_1^T)}{2h}$ of largest magnitude. Assume it is
entry $\{i,j\}$. Then determine entries of largest magnitude of columns $i$ and $j$ of the
matrix $\U_{\ell}$. The product of these entries yields an approximation of an entry of
largest magnitude of $L_f(\A^T,ee^T)$.

\section{Network modifications based on Perron root sensitivity}\label{sec3} 
The methods of this section require right and left Perron vectors of the adjacency matrix 
$\A$. When the matrix $\A$ is of small to moderate size, these vectors can be determined 
with the MATLAB function ${\sf eig}$, which computes all eigenvalues and eigenvectors of
$\A$. For large networks, we can compute the Perron vectors with the MATLAB function 
{\sf eigs} or with the two-sided Arnoldi method. The latter method was introduced by Ruhe
\cite{Ru} and improved by Zwaan and Hochstenbach \cite{ZH}. 

Our interest in the method of this section stems from the fact that it is easy to 
implement because the required computations are quite straightforward. However, the method
does not identify edge weights whose modification yields a relatively large change in the 
total communicability \eqref{TC}. Instead, it identifies edge weights whose modification 
gives a relatively large change in the Perron root of the adjacency matrix. Computed
examples in Section \ref{sec4} indicate that modifications of edge weights identified by 
this method also results in relatively large changes in the total communicability.

\subsection{Perron communicability for small to medium-sized networks}\label{subsec3.1}
Let the adjacency matrix $\A$ for the graph ${\mathcal G}$ be irreducible and let $\rho$
be its Perron root. Then there are unique right and left eigenvectors 
$x = [x_1,x_2,\ldots,x_n]^T\in\R^n$ and $y = [y_1,y_2,\ldots,y_n]^T\in\R^n$, respectively, 
of unit Euclidean norm with positive entries associated with $\rho$, i.e.,
\[
\A x = \rho x,\qquad  y^T\A = \rho y^T. 
\]
They are referred to as Perron vectors. Let $\F\in\R^{n\times n}$ be a nonnegative matrix 
of unit spectral norm, $\|\F\|_2=1$. Introduce the small positive parameter 
$\varepsilon$ and denote the Perron root of $\A + \varepsilon \F$ by $\rho+\delta\rho$. 
Then
\[
\delta\rho=\varepsilon \frac{y^T\F x}{y^Tx}+{\mathcal O}(\varepsilon^2)
\]
and
\begin{equation}\label{W}
\frac{y^T\F x}{y^Tx}\leq\frac{\|y\|_2\|\F\|_2\|x\|_2}{y^Tx}=\frac{1}{\cos\theta},
\end{equation}
where $\theta$ is the angle between $x$ and $y$. The quantity $1/\cos\theta$ is referred 
to as the {\it condition number} of $\rho$ and denoted by $\kappa(\rho)$; see 
\cite[Section 2]{Wi}. Note that when $\A$  is symmetric, we have $x=y$, hence $\theta=0$.
Equality  in \eqref{W} is attained when $\F$ is the {\it Wilkinson perturbation} $\W=yx^T$
associated with $\rho$; see \cite{MSN,Wi} for details.

The total communicability \eqref{TC} of the graph ${\mathcal G}$ can be approximated 
by the {\it Perron communicability} of ${\mathcal G}$ \cite{DLCJNR}:
\begin{equation}\label{PC}
P_{\mathcal G}(w_{11},w_{21},\ldots,w_{nn})=\exp(\rho) e^Tyx^T e= \exp(\rho) e^T\W e
\end{equation}
with 
\begin{equation*}
T_{\mathcal G}(w_{11},w_{21},\ldots,w_{nn})\approx 
\kappa(\rho)\,P_{\mathcal G}(w_{11},w_{21},\ldots,w_{nn}).
\end{equation*}
Typically, $\exp(\rho)$ is a fairly accurate indicator of the Perron communicability 
 and, consequently, of the total communicability. In fact, one has \cite{DLCJNR}:
\begin{equation}\label{PCbds}
\exp(\rho) \leq P_{\mathcal G}(w_{11},w_{21},\ldots,w_{nn}) \leq n\exp(\rho).
\end{equation}

Perturb the entry $w_{ij}$ with $i\ne j$ of $\A$ by $\varepsilon\ne 0$ and let 
\begin{equation}\label{Apert}
\F=e_ie_j^T\in{\mathcal{A}}
\end{equation}
for some index pair $\{i,j\}$. The perturbation $\delta\rho$ of $\rho$ due to the 
perturbation $\varepsilon\F$ of $\A$ is
\begin{equation}\label{deltarho}
\delta\rho=\varepsilon\frac{y_ix_j}{y^Tx}+{\mathcal O}(\varepsilon^2). 
\end{equation}

\subsubsection{Network simplification by edge removal}\label{subsubsec311}
To  reduce the complexity of a network by removing edges without changing the Perron 
communicability significantly, we choose the matrix 
\eqref{Apert} so that $\rho$ (and hence $\exp(\rho)$)  changes as little as possible 
and, therefore, choose the indices $i$ and $j$ so that
\[
w_{ij}y_ix_j=\min_{\substack{1\leq h,k\leq n\\ w_{hk}>0}}w_{hk}y_hx_k,
\]
and use $\A-\varepsilon \F$ with $\varepsilon=w_{ij}$ and $\F$ given by \eqref{Apert},

If the graph is undirected, then we choose the matrix
\[
\F=\frac{e_ie_j^T+e_je_i^T}{2}\in{\mathcal{A}},
\]
with the indices $i$ and $j$ determined as above, and use $\A-\varepsilon \F$ with
$\varepsilon=2\,w_{ij}=2\,w_{ji}$.

Introduce the vector $ {\mathcal{E}}_{\rho} \in \R^m$, whose $k$th entry is the 
{\it Perron edge importance} of the edge $e_k=e(v_i\to v_j)$, defined as the product 
of the edge-weight $w_{ij}$ and the corresponding entry $y_ix_j$ of $\W$.  Observe 
that  $||\W ||_F=||\W||_2=1$. The procedure to construct the Perron edge importance vector
${\mathcal{E}}_{\rho}$ consists of two steps:\\

\noindent
{\bf Procedure 5}:
\begin{enumerate}
\item Multiply the adjacency matrix $\A$ element by element by $\W|_{\mathcal{A}}$.
\item Put column by column the $m$ nonvanishing entries of the matrix so obtained into 
the vector ${\mathcal{E}}_{\rho}$.
\end{enumerate}

If the graph is undirected, then the {\it Perron edge importance} of  edge 
$e_k=e(v_i\leftrightarrow v_j)$ is defined as twice the product of the edge-weight 
$w_{ij}$ and the corresponding entry $y_ix_j$ of $\W$, so that the procedure 
to construct the Perron edge importance vector ${\mathcal{E}}_{\rho}$ becomes: \\

\noindent
{\bf Procedure 6}:
\begin{enumerate}
\item Multiply $\A|_{\mathcal{L}}$ element by element by $\W|_{\mathcal{L}}$.
\item Multiply by $2$ the $m$ nonvanishing entries of the matrix so obtained and
put them column by column into the vector ${\mathcal{E}}_{\rho}$.
\end{enumerate}

In computations, we order the entries of ${\mathcal{E}}_{\rho}$ from smallest to largest,
and set the entries $w_{ij}$ (or the pair of entries  $w_{ij}=w_{ji}$ if the graph is undirected) 
associated with the first few of the ordered edge importances to zero. As mentioned before
some post-processing may be necessary if the reduced graph is required to be connected.

\subsubsection{Network modification by edge-weight tuning}\label{subsubsec312}
We describe how to increase the total communicability and use the notation of 
Subsection \ref{subsubsec311}. The discussion follows \cite{NR}. 
We would like to choose a perturbation $\varepsilon\F$ of $\A$, where $\varepsilon>0$ and
$\F$ is of the form \eqref{Apert}, so that the Perron root $\rho$ increases as much as 
possible. This suggests that we choose the indices $i$ and $j$ in \eqref{Apert} so that
\[
w_{ij}y_ix_j=\max_{\substack{1\leq h,k\leq n, \\
w_{hk}>0}}w_{hk}y_hx_k.
\]
Thus, we choose the weight associated with the largest entry of the vector 
${\mathcal{E}}_{\rho}$ that yields the Perron edge importance of each edge. 

We turn to the reduction of the total communicability. Define the matrix $\F$ as above
and consider the perturbed matrix $\A-\varepsilon\F$. The parameter $\varepsilon>0$ 
should be chosen small enough so that this matrix has nonnegative entries only. Moreover,
if removing an edge $e_k=e(v_i\to v_j)$ associated with one of the first few of the 
ordered entries of ${\mathcal{E}}_{\rho}$ results in a disconnected graph and this is 
undesirable, then we choose $\varepsilon>0$ so that $0<\varepsilon<w_{ij}$. Analogously, if
removing the edges $e(v_i\leftrightarrow v_j)$ of an undirected graph ${\mathcal G}$ makes
the graph disconnected and this is undesirable, then we choose $\varepsilon>0$ so that
$\varepsilon<2\,w_{ij}=2\,w_{ji}$.

\subsubsection{Network modification by inclusion of new edges}\label{subsubsec313}
Let $\F\in{\mathcal{A}}$ be a nonnegative matrix with
$\|\F\|_F=1$, and let $\varepsilon>0$ be a small constant. Then
\[
\frac{y^T\F x}{y^Tx}\leq\frac{\|y\|_2\| \|yx^T|_{\mathcal{A}}\|_F\|x\|_2}{y^Tx}=
\frac{\|\W|_{\cal{A}}\|_F}{\cos\theta},
\]
with equality for the ${\mathcal{A}}$-\emph{structured analogue of the Wilkinson 
perturbation},
\[
\F=\frac{ \W|_{\mathcal{A}} }{ \|\W|_{\mathcal{A}}\|_F }.
\]
This is the maximal perturbation for the Perron root $\rho$ induced by a unit norm matrix 
$\F\in{\mathcal{A}}$; see \cite{EHNR,NP}. The quantity
\[
\frac{\|\W|_{\mathcal{A}}\|_F}{\cos\theta}=\kappa(\rho)\|\W|_{\cal{A}}\|_F
\]
is referred to as the ${\mathcal{A}}$-\emph{structured condition number} of $\rho$ and 
denoted by $\kappa_{{\mathcal{A}}}(\rho)$. Thus, 
$\kappa_{{\mathcal{A}}}(\rho)\leq\kappa(\rho)$.

To increase the Perron communicability, we would like to modify the edges of the graph 
${\mathcal G}$ so that the Perron root is increased as much as possible; cf.
\eqref{PCbds}. In case the $m$ edges of ${\mathcal{G}}$ are such that 
\begin{equation}\label{Wprop}
\|\W|_{\mathcal{A}}\|_F\approx \|\W\|_F=1, 
\end{equation}
i.e., when $\kappa_{{\mathcal{A}}}(\rho) \approx \kappa(\rho)$, increasing positive 
entries of $\A$ should be a successful strategy to increase the Perron communicability.
In fact, the matrix $\mathbf{S}=[s_{ij}]_{i,j=1}^n\in {\cal{A}}$, with entries 
$s_{ij}=\frac{y_ix_j}{y^Tx}$, if $w_{ij}>0$ and $s_{ij}=0$ otherwise, referred to as
the {\it structured Perron sensitivity matrix},  is such that
$$\mathbf{S}=\kappa(\rho)\W|_{\mathcal{A}}=\kappa(\rho)\|\W|_{\cal{A}}\|_F
\frac{\W|_{\cal{A}}}{\|\W|_{\cal{A}}\|_F},$$
so that $\|\mathbf{S}\|_F=\kappa_{{\mathcal{A}}}(\rho)\approx \kappa(\rho)$.
 If $\F$ is of the form \eqref{Apert}, the perturbation  \eqref{deltarho} of $\rho$ 
induced by $\varepsilon\F$ can be written as $\delta\rho=\varepsilon s_{ij}+
{\mathcal O}(\varepsilon^2)$.

Conversely, if $\kappa_{{\mathcal{A}}}(\rho) \ll \kappa(\rho)$, then the addition of
a suitable edge with weight $w_{ij}>0$ (or a suitable pair of edges with weights 
$w_{ij}=w_{ji}>0$, if the graph is undirected) that increases the ratio
$\kappa_{{\mathcal{A}}}(\rho)/\kappa(\rho)$ may be an appropriate strategy to increase 
the Perron communicability.
Recall that $\widehat{\mathcal{A}}$ denotes the cone of the nonnegative matrices in 
$\R^{n\times n}$ whose sparsity structure is given by the zero entries of $\A$ except for the 
diagonal entries.  
Perturb the entry $w_{ij}$ with $i\ne j$ of $\A$  and assume that
\begin{equation*}
\F=e_ie_j^T\in\widehat{\mathcal{A}}
\end{equation*}
for some index pair $\{i,j\}$. The entries of the Wilkinson perturbation $\W$ reveal which edges 
should be added to the network to increase the communicability, namely edges whose 
associated entries of the matrix $\W$ are large. The procedure for the construction of the
vector $\widehat{{\mathcal{E}}}_{\rho}\in \R^{n^2-n-m}$ that gives the {\it Perron virtual 
importance} of the virtual edges is the following:\\

\noindent
{\bf Procedure 7}:
\begin{enumerate}
\item Construct the matrix $\W|_{\widehat{\mathcal{A}}}$.
\item Put column by column the entries of the matrix so obtained that belong to
the sparsity structure associated with $\widehat{\mathcal{A}}$ into the vector
$\widehat{{\mathcal{E}}}_{\rho}$.
\end{enumerate}

If the graph is undirected, then the Perron virtual importance of the nonexistent edge 
$e(v_i\leftrightarrow v_j)\notin {\mathcal E}$ is defined to be twice the corresponding 
entry in $\W$. The procedure for the construction of the Perron virtual edge importance 
vector $\widehat{{\mathcal{E}}}_{\rho}\in \R^{(n^2-n-2m)/2}$ is given by:\\

\noindent
{\bf Procedure 8}:
\begin{enumerate}
\item Construct the matrix $\W|_{\widehat{\mathcal{L}}}$.
\item Multiply by $2$ the entries of the matrix so obtained that belong to  the sparsity 
structure associated with $\widehat{\mathcal{L}}$ and put them column by column into the
vector $\widehat{{\mathcal{E}}}_{\rho}$.
\end{enumerate}

\subsection{Network modification criteria for large-scale networks}\label{subsec32}
Introduce the {\it structured Perron communicability} of ${\mathcal G}$:
\[
P_{\mathcal G}^{{\mathcal{A}}} (w_{11},w_{21},\ldots,w_{nn})= \exp(\rho) e^T
\W|_{{\mathcal{A}}} \,e.
\]
One has, entry-wise,  $\W|_{{\mathcal{A}}} \leq \W$, so that the structured Perron 
communicability $P_{\mathcal G}^{{\mathcal{A}}}(w_{11},w_{21},\ldots,w_{nn})$ is a lower
bound for the Perron communicability \eqref{PC}. When 
$\kappa_{{\mathcal{A}}}(\rho) \approx \kappa(\rho)$, i.e., when \eqref{Wprop} holds, the
two measures are very close. 

Additionally, if ${\mathcal G}$ is undirected, then one has 
\begin{equation*}
P_{\mathcal G}^{{\mathcal{A}}} (w_{11},w_{21},\ldots,w_{nn})= 2\exp(\rho) e^T
\W|_{{\mathcal{L}}} \,e,
\end{equation*}
with ${\mathcal{L}}$ the cone of all nonnegative matrices in $\R^{n\times n}$ with 
the same sparsity structure as the strictly lower triangular portion of $\A$.

If our aim is to perturb or set to zero suitable positive entries of $\A\in{\mathcal{A}}$, 
then the Wilkinson perturbation $\W$ does not have to be constructed, since one only needs
the entries of $\W|_{\mathcal{A}}\in{\mathcal{A}}$. 
The Perron edge importance vector ${\mathcal{E}}_{\rho} \in \R^{m}$, associated with the $m$ 
edges of ${\mathcal G}$, can be evaluated as discussed in Subsection \ref{subsubsec311}.

\section{Computed examples} \label{sec4}
The numerical tests reported in this section have been carried out using MATLAB R2023a on 
a $3.2$ GHz Intel Core i7 6 core iMac. 

\subsection{A synthetic example}
The following example explains why we might be  interested in estimating the total communicability 
without evaluating the exponential  of the adjacency matrix associated with the given network.
\begin{example}\label{ex0}
Consider the  undirected and connected graph where each vertex represents a person in a line. 
Each person can communicate only with the following and the preceding persons.
The adjacency matrix associated with such a network is the symmetric tridiagonal Toeplitz
matrix
\begin{equation}\label{Atoep}
\A=\left[ 
\begin{array}{cccccccc}
\phantom{vv} & \sigma &  &  &  &  &  &\\ 
\sigma & \phantom{vv} & \sigma &  &  &  &  &\\ 
&  \sigma & \phantom{vv}  & \cdot &  &  &  &\\ 
&  &   \cdot &  \phantom{vv} & \cdot & &  \\ 
&  &    & \cdot & \phantom{vv} & \sigma \\ 
 &  &  &  & \sigma & \phantom{vv} 
\end{array}
\right]\in\R^{n\times n},
\end{equation}
with $\sigma>0$. It is intuitive that the closer the vertices are to the center of the graph, the better 
connected they are and, therefore, the more important.
Also, it is immediate to observe that the best strategy to improve communication in this 
network is to make the two vertices at the ends adjacent by adding the new undirected edge 
$e(v_1\leftrightarrow v_n)$ with weight $\sigma$, which makes the adjacency matrix 
a circulant. This scenario is consistent with the structure of the Wilkinson matrix 
$\W$ associated with the Perron root of $\A$, which we show in  Figure \ref{FIG0}(b). This
matrix is independent of $\sigma$;  see  \cite[Proposition 1]{NP}. 
Therefore the network modifications based on the Perron root sensitivity described in 
Section \ref{sec3}
are in agreement with the previous observations and considerations.

Instead, due to round-off errors introduced in the representation and evaluation of the 
matrix exponential, we have that the
the structure of the matrix $L_{\exp_0}(\A^T,ee^T)$  (see Figure \ref{FIG0}(a)) does not 
completely reflect the scenario described above. In fact,  the techniques  based on 
the gradient \eqref{grad} described in Section \ref{sec2} do not identify the same edge 
as the Wilkinson perturbation. For instance, let
$n$ be even and $\sigma=1$. Then we obtain for $n\geq 36$ that the computed gradient does not
identify  the undirected edge $e(v_{n/2}\leftrightarrow v_{n/2+1})$ as the most important one.

\begin{figure}[h!tb]
\centering
\begin{tabular}{cc}
\includegraphics[scale = 0.40]{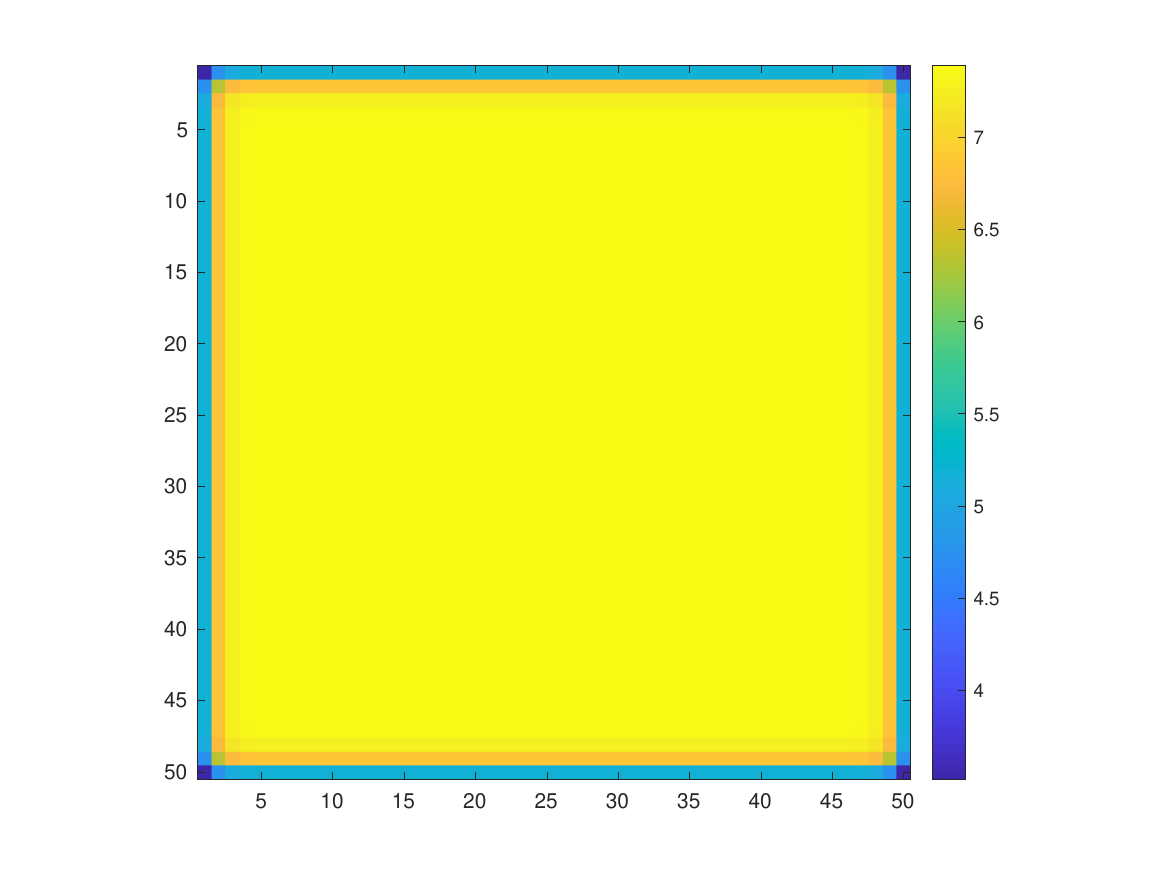} & 
\includegraphics[scale = 0.40]{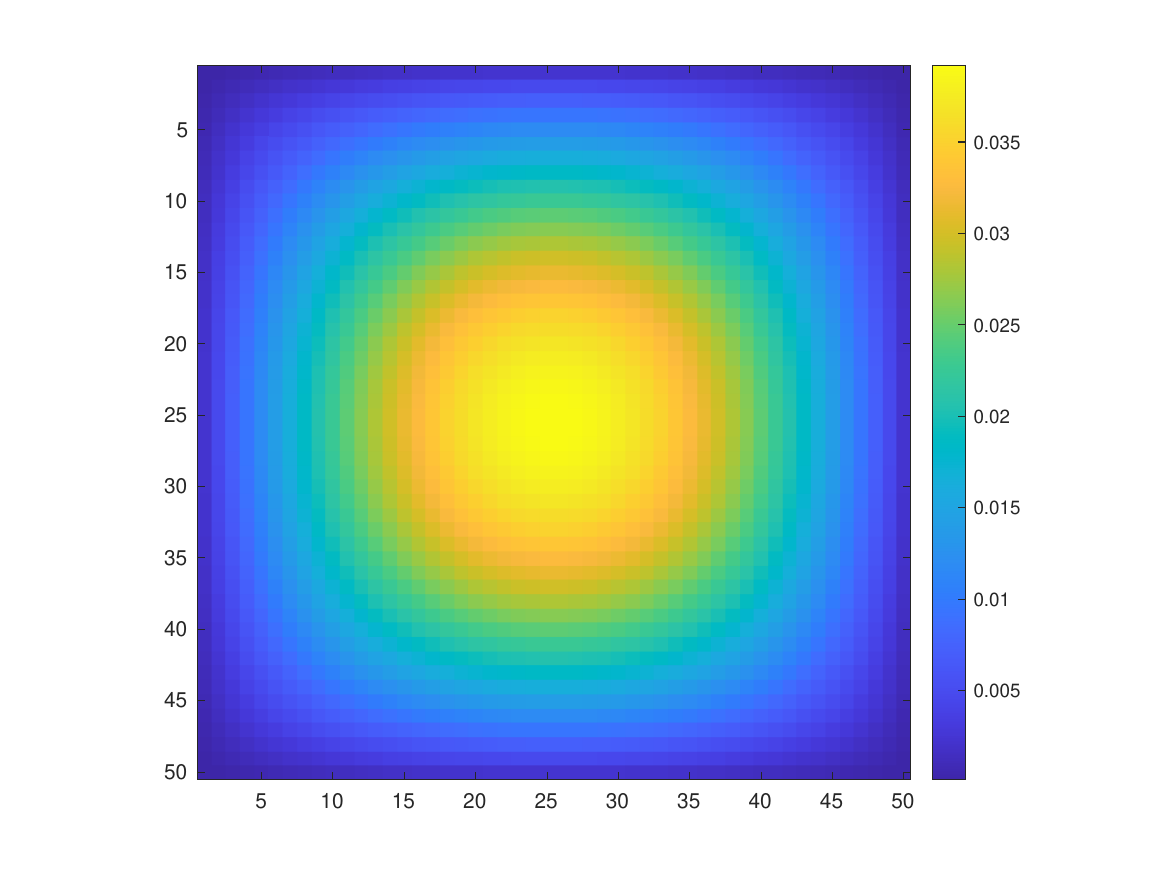} \\
 (a) & (b) 
\end{tabular}
\caption{Example \ref{ex0}. Structure of $L_{\exp_0}(\A,ee^T)$ (left picture (a))  and of $\W$ 
(right picture (b)) for the matrix $\A\in \R^{50\times 50}$ in \eqref{Atoep}
with $\sigma=1$.}
\label{FIG0}
\end{figure}
\end{example}

\subsection{Medium-sized networks}

\begin{example}\label{ex1}
Consider the adjacency matrix $\A\in\R^{500\times 500}$ for the network {\em Air500} 
in \cite{air500_autobahn}. This data set describes flight connections for the top 500 airports 
worldwide based on total passenger volume. The flight connections between airports are for 
the year from 1 July 2007 to 30 June 2008, and the network is represented by a directed 
unweighted connected graph ${\mathcal{G}}$ with $n=500$ vertices and $m=24009$ directed 
edges. Vertices of the network are airports and edges represent direct flight routes 
between two airports. 

The total communicability  in \eqref{TC} is $T_{\mathcal G}=1.9164\cdot 10^{38}$.
The gradient $\nabla T_{\mathcal G}$ in \eqref{grad} has been computed by evaluating 
the matrix $L_f(\A^T,ee^T)$ in \eqref{Lfexpan} using \eqref{Lf}. 
The total transmission is $\|\nabla T_{\mathcal G}\|_2=1.9205\cdot 10^{38}$. Also, the 
gradient $\nabla T_{\mathcal G}$ has been approximated by evaluating $L_f(\A^T,ee^T)$ 
using \eqref{Lfapprox} with $h=\frac{2}{n}\cdot10^{-4}=4\cdot10^{-7}$, obtaining 
$\widetilde{\nabla T}_{\mathcal G}$.  The resulting total transmission is 
$\|\widetilde{\nabla T}_{\mathcal G}\|_2=1.9205\cdot 10^{38}$ with 
\[
\frac{\|\nabla T_{\mathcal G}-\widetilde{\nabla T}_{\mathcal G} \|_2}
{{\|\nabla T_{\mathcal G}}\|_2}=1.9688\cdot 10^{-9}.
\]
For both the edge importance vector  ${\mathcal{E}}_{L_{f}}$ and the virtual 
edge importance vector $\widehat{\mathcal{E}}_{L_{f}}$, the same results 
are obtained regardless of whether  $\nabla T_{\mathcal G}$ or 
$\widetilde{\nabla T}_{\mathcal G}$ is used; see below. We  remark that evaluating $L_f(\A^T,ee^T)$ by \eqref{Lf}
required the average time $t_1\approx 3.4\cdot 10^{-1}$ in $10^4$ tests, while evaluating
$L_f(\A^T,ee^T)$ using \eqref{Lfapprox} 
required the average time $t_2\approx 1.5\cdot 10^{-1}$ in the same tests, with a 
relative average time saving
$$
\frac{t_1-t_2}{t_1}\approx 5.6\cdot 10^{-1}.
$$

The Perron communicability  in \eqref{PC} is $P_{\mathcal G}=1.9132\cdot 10^{38}$. 
The Perron root and left and right Perron vectors were evaluated with the 
MATLAB function {\sf eig}.

\begin{table}[h!tb]
\centering
\begin{tabular}{c l |c r}
${\mathcal{E}}_{L_{f}}$          & $e(v_i\to v_j)$ & ${\mathcal{E}}_{\rho}$ & $e(v_i\to v_j)$  \\
\hline
$9.0980\cdot 10^{-8}$  & \bf{TSA}  $\to$ MZG  & $3.1829\cdot 10^{-9}$ & \bf{TSA}  $\to$ MZG   \\
$9.0980\cdot 10^{-8}$  & \bf{MZG} $\to$ TSA   & $3.1829\cdot 10^{-9}$ & \bf{MZG}  $\to$ TSA   \\
$5.1931\cdot 10^{-7}$  & UKB  $\to$ ISG   & $1.5705\cdot 10^{-8}$ & SDU  $\to$ CGH   \\
$5.8446\cdot 10^{-7}$  & ISG   $\to$ UKB  & $1.5705\cdot 10^{-8}$ & CGH  $\to$ SDU   \\
$8.8771\cdot 10^{-7}$  & SDU $\to$ CGH  & $6.9399\cdot 10^{-8}$ & UKB  $\to$ ISG   \\
$9.2224\cdot 10^{-7}$  & CGH $\to$ SDU  & $8.5792\cdot 10^{-8}$ & ISG  $\to$ UKB   \\
$9.6419\cdot 10^{-7}$  & GMP $\to$ HND  & $1.2677\cdot 10^{-7}$ & HND  $\to$ GMP   \\
$9.7169\cdot 10^{-7}$  & HND $\to$ GMP  & $1.2979\cdot 10^{-7}$ & GMP  $\to$ HND   \\      
$1.0369\cdot 10^{-6}$  & DUR $\to$ PLZ   & $1.7020\cdot 10^{-7}$ &  UKB  $\to$ HND  \\
$1.0376\cdot 10^{-6}$  & PLZ  $\to$ DUR  & $1.9511\cdot  10^{-7}$ & HND  $\to$ UKB 
\end{tabular}
\caption{Example \ref{ex1}. The first $10$ flight connections to remove in order to reduce 
the complexity of {\em Air500} without changing the network communicability significantly, 
according to the determination of the edge importance based on gradient and on Perron root  
sensitivity, respectively.} 
\label{T_1}
\end{table}

In Table \ref{T_1}  the $10$ smallest entries of both the edge importance vector 
${\mathcal{E}}_{L_{f}}$  and the Perron edge importance vector ${\mathcal{E}}_{\rho}$  
are shown along with the corresponding edges. This is useful for determining which 
edges to remove in order to reduce the complexity of the {\em Air500} network (cf.
Subsections \ref{subsubsec211} and \ref{subsubsec311}). 
One can observe that the elimination of the connection from Santos Dumont Airport,
Rio de Janeiro, Brazil (SDU)  to Congonhas Airport, S. Paulo, Brazil (CGH), at the third 
position in the ranking given by ${\mathcal{E}}_{\rho}$, would disconnect the network. We
therefore only remove the two edges in bold face in Table \ref{T_1}, 
i.e., the two most irrelevant edges - according to edge importance determination based  
on both gradient and Perron root  sensitivity. These edges correspond to the flight connection 
between San Antonio International Airport, San Antonio, Texas  (TSA)  and Penghu Airport, 
Taiwan (MZG). This results in  the network ${\mathcal G}_1$,  for which we have:
\[
T_{{\mathcal G}_1}= 1.9164\cdot 10^{38};\;\;\;\;
\frac{ T_{\mathcal G}-T_{{\mathcal G}_1} } {T_{\mathcal G}} =
1.8014\cdot 10^{-7} ;\;\;\;\;\;\;\;\;
P_{{\mathcal G}_1}= 1.9132\cdot 10^{38};\;\;\;\;
\frac{ P_{\mathcal G}-P_{{\mathcal G}_1} }{P_{\mathcal G}} =
1.8014\cdot 10^{-7}.
\]

\begin{table}[h!tb]
\centering
\begin{tabular}{c l |c r}
${\mathcal{E}}_{L_{f}}$          & $e(v_i\to v_j)$ & ${\mathcal{E}}_{\rho}$ & $e(v_i\to v_j)$  \\
\hline 
$1.9810\cdot 10^{-2}$  & \bf{JFK}  $\to$ ATL  & $2.0038\cdot 10^{-2}$ & \bf{JFK}  $\to$ ATL   \\
$1.9757\cdot 10^{-2}$  & \bf{ORD} $\to$ JFK   & $1.9987\cdot 10^{-2}$ & \bf{ORD}  $\to$ JFK   \\
$1.9660\cdot 10^{-2}$  & JFK  $\to$ ORD  & $1.9882\cdot 10^{-2}$ & JFK  $\to$ ORD   \\
$1.9625\cdot 10^{-2}$  & ATL   $\to$ JFK  & $1.9861\cdot 10^{-2}$ & ATL  $\to$ JFK   \\
$1.9152\cdot 10^{-2}$  & JFK $\to$ LAX  & $ 1.9369\cdot 10^{-2}$ & JFK  $\to$ LAX   \\
$1.9068\cdot 10^{-2}$  & EWR $\to$ JFK  & $1.9280\cdot 10^{-2}$ & EWR  $\to$ JFK  \\
$1.8959\cdot 10^{-2}$  & JFK $\to$ EWR & $1.9239\cdot 10^{-2}$ & ORD  $\to$ ATL \\
$1.8945\cdot 10^{-2}$  & ORD $\to$ ATL  & $1.9165\cdot 10^{-2}$ & JFK  $\to$ EWR   \\      
$1.8727\cdot 10^{-2}$  & LAX $\to$ JFK   & $1.8970\cdot 10^{-2}$ &  ATL  $\to$ ORD  \\
$1.8677\cdot 10^{-2}$  & ATL  $\to$ ORD  & $1.8936\cdot  10^{-2}$ & LAX  $\to$ JFK 
\end{tabular}
\caption{Example \ref{ex1}. The first $10$ flight connections to increase/decrease in order 
to increase/decrease the network communicability in {\em Air500} according to the determination 
of the edge importance based on gradient and on Perron root  sensitivity, respectively.} 
\label{T_2}
\end{table}

Table \ref{T_2} shows the $10$ largest entries of both the edge importance vector 
${\mathcal{E}}_{L_{f}}$ and the Perron edge importance vector ${\mathcal{E}}_{\rho}$
along with the corresponding edges. 
In order to obtain a relatively large reduction in the total communicability, we set to 
zero the weights associated with the two largest entries of both  ${\mathcal{E}}_{L_{f}}$
and ${\mathcal{E}}_{\rho}$. This means we remove the edges that represent air route 
between John F. Kennedy 
International  Airport, New York City (JFK) and Atlanta Hartsfield-Jackson Airport,  
 Georgia (ATL).
This results in the network ${\mathcal G}_2$,  for which, as it is apparent,
the reduction in the total communicability is much larger than in ${\mathcal G}_1$:
\[
T_{{\mathcal G}_2}= 1.8423\cdot 10^{38};\;\;\;\;
\frac{ T_{\mathcal G}-T_{{\mathcal G}_2} } {T_{\mathcal G}} =
3.8680\cdot 10^{-2} ;\;\;\;\;\;\;\;\;
P_{{\mathcal G}_2}= 1.8392\cdot 10^{38};\;\;\;\;
\frac{ P_{\mathcal G}-P_{{\mathcal G}_2} }{P_{\mathcal G}} =
3.8689\cdot 10^{-2}.
\]
Following the discussion in Subsections \ref{subsubsec212} and \ref{subsubsec312}, in order to 
obtain a relatively large increase in the total communicability, we  increase by $1$ 
the edge-weights associated with the two largest entries of  both ${\mathcal{E}}_{L_{f}}$
and ${\mathcal{E}}_{\rho}$ (i.e.,  the air route between John F. Kennedy 
International  Airport, New York City (JFK), and Atlanta Hartsfield-Jackson Airport,  
 Georgia (ATL)). For the so obtained network ${\mathcal G}_3$,  one has
\[
T_{{\mathcal G}_3}= 1.9943\cdot 10^{38};\;\;\;\;
\frac{ T_{{\mathcal G}_3}-T_{\mathcal G}} {T_{\mathcal G}} =
4.0658\cdot 10^{-2} ;\;\;\;\;\;\;\;\;
P_{{\mathcal G}_3}= 1.9910\cdot 10^{38};\;\;\;\;
\frac{ P_{{\mathcal G}_3}-P_{\mathcal G}} {P_{\mathcal G}} =
4.0661\cdot 10^{-2}.
\]
 \begin{table}[h!tb]
\centering
\begin{tabular}{c l |c r}
$\widehat{\mathcal{E}}_{L_{f}}$& $e(v_i\to v_j)$ & $\widehat{{\mathcal{E}}}_{\rho}$ & $e(v_i\to v_j)$  \\
\hline
$1.3153\cdot 10^{-2}$  & JFK  $\to$ LGA  & $1.3385\cdot 10^{-2}$ & JFK  $\to$ LGA   \\
$1.3102\cdot 10^{-2}$  & LGA $\to$ JFK    & $1.3328\cdot 10^{-2}$ & LGA $\to$ JFK   \\
$1.2402\cdot 10^{-2}$  & \bf{LHR}  $\to$ ATL  & $1.2383\cdot 10^{-2}$ & \bf{MDW}  $\to$ JFK   \\
$1.2276\cdot 10^{-2}$  & \bf{AMS}   $\to$ DFW  & $1.2278\cdot 10^{-2}$ & \bf{LHR}  $\to$ ATL   \\
$1.2252\cdot 10^{-2}$  & ATL $\to$ LHR  & $1.2264\cdot 10^{-2}$ & JFK  $\to$ MDW   \\
$1.2160\cdot 10^{-2}$  & MDW $\to$ JFK  & $1.2212\cdot 10^{-2}$ & AMS  $\to$ DFW  \\
$1.2039\cdot 10^{-2}$  & JFK $\to$ MDW  & $1.2136\cdot 10^{-2}$ & ABQ  $\to$ JFK  \\
$1.1915\cdot 10^{-2}$  & ABQ $\to$ JFK  & $1.2093\cdot 10^{-2}$ & ATL  $\to$ LHR   \\      
$1.1581\cdot 10^{-2}$  & DFW $\to$ AMS   & $1.1775\cdot 10^{-2}$ &  ORD  $\to$ MDW  \\
$1.1525\cdot 10^{-2}$  & ORD  $\to$ LGW  & $1.1472\cdot  10^{-2}$ & DFW  $\to$ AMS 
\end{tabular}
\caption{Example \ref{ex1}. The first $10$ flight connections that should be added, in order to 
enhance the communicability in {\em Air500}, according to determination of the 
virtual edge importance based on gradient and on Perron root  sensitivity, respectively.} 
\label{T_3}
\end{table}

Finally, following the discussion in Subsections \ref{subsec5} and \ref{subsubsec313}, we
display in Table \ref{T_3} the $10$ largest entries of the total virtual  edge importance 
vector $\widehat{\mathcal{E}}_{L_{f}}$ and the $10$ largest entries of the Perron virtual 
edge importance vector $\widehat{{\mathcal{E}}}_{\rho}$, along with the corresponding 
{\it nonexistent} edges. Notice that the edges associated with the two largest entries of  
both  $\widehat{\mathcal{E}}_{L_{f}}$ and $\widehat{{\mathcal{E}}}_{\rho}$ cannot be
considered because they model  the missing air route between John F. Kennedy 
International  Airport, New York City, (JFK), and La Guardia Airport, New York City,
(LGA), and are too close to justify a flight route. The entries of the table suggest that
there should be a shuttle service between these vertices and, indeed, such a shuttle service
exists. We proceed to  consider the third and fourth best nonexistent edges according to 
the edge importance based on the gradient, that is we consider the routes from Heathrow 
Airport, London, England, (LHR) to Atlanta  Hartsfield-Jackson Airport, Georgia, (ATL) and
from  the Amsterdam Schiphol Airport, Netherlands, (AMS) to Dallas/Fort Worth International
Airport, Texas, (DFW), and set their weights to $1$.  This way, we obtain the network 
${\mathcal G}_4$, for which one has
\[
T_{{\mathcal G}_4}= 1.9643\cdot 10^{38};\;\;\;\;
\frac{ T_{{\mathcal G}_4}-T_{{\mathcal G}} } {T_{\mathcal G}} =
2.5020\cdot 10^{-2}.
\]

Conversely, setting to $1$ the weight of the third and the fourth best nonexistent edges 
according to $\widehat{{\mathcal{E}}}_{\rho}$, i.e., the edges that represent flights from 
Midway International Airport, Chicago, Illinois, (MDW) to John F. Kennedy International 
Airport, New York City, (JFK) and from Heathrow Airport, London, England, (LHR) to 
Hartsfield-Jackson Airport, Atlanta, Georgia (ATL), we obtain the network ${\mathcal G}_5$
with 
\[
P_{{\mathcal G}_5}=1.9609\cdot 10^{38};\;\;\;\;
\frac{ P_{{\mathcal G}_5}-P_{{\mathcal G}} } {P_{\mathcal G}} =
2.4915\cdot 10^{-2};\;\;\;\;\;\;\;\;
T_{{\mathcal G}_5}= 1.9641\cdot 10^{38};\;\;\;\;
\frac{ T_{{\mathcal G}_5}-T_{{\mathcal G}} } {T_{\mathcal G}} =
2.4908\cdot 10^{-2}.
\]
In this example,  edge addition is less effective than increasing the weights of  existing 
edges.  We observe that the matrix in $\mathcal{A}$  that is closest to the Wilkinson 
perturbation $\W$ associated to the Perron root  $\rho$ of $\A$ with respect to the Frobenius 
norm, i.e., $\W|_{\mathcal{A}}$, has Frobenius norm  $\|\W|_{\mathcal{A}}\|_F
=7.5920\cdot 10^{-1}$, meaning that  the $\mathcal{A}$-structured condition number 
$\kappa_{{\mathcal{A}}}(\rho)$ is approximately $76\%$ of the condition number $\kappa(\rho)$.
\end{example}

\begin{example}\label{ex2}
Consider the undirected unweighted graph $\mathcal G$ that represents the German 
highway system network {\em Autobahn}. The graph, which is available at 
\cite{air500_autobahn}, has $n=1168$ vertices representing German locations and $m=1243$ 
edges that represent highway segments connecting them. Therefore, the adjacency 
matrix $\A\in\R^{1168\times 1168}$ for this network has $2486$ nonvanishing entries. 

The total communicability  in \eqref{TC} is $T_{\mathcal G}=1.2563\cdot 10^{4}$.
The gradient $\nabla T_{\mathcal G}$ in \eqref{grad} has been computed by evaluating 
the matrix $L_f(\A^T,ee^T)$ in \eqref{Lfexpan} using \eqref{Lf}. 
The total transmission is $\|\nabla T_{\mathcal G}\|_2=1.4464\cdot 10^{4}$. The 
gradient $\nabla T_{\mathcal G}$ has been approximated by evaluating $L_f(\A^T,ee^T)$ 
using \eqref{Lfapprox} with $h=\frac{2}{n}\cdot10^{-4}=1.7123\cdot 10^{-7}$, obtaining 
$\widetilde{\nabla T}_{\mathcal G}$.  The resulting total transmission is 
$\|\widetilde{\nabla T}_{\mathcal G}\|_2=1.4464\cdot 10^{4}$ with 
\[
\frac{\|\nabla T_{\mathcal G}-\widetilde{\nabla T}_{\mathcal G} \|_2}
{{\|\nabla T_{\mathcal G}}\|_2}=5.4537\cdot 10^{-9}.
\]
Notice that evaluating the matrix $L_f(\A^T,ee^T)$ in $10^4$ tests required the
average time $t_1\approx 1.8\cdot 10^{0}$ using  \eqref{Lf} and  the average time 
$t_2\approx 3.2\cdot 10^{-1}$ using \eqref{Lfapprox}, with a relative average 
time saving
$$
\frac{t_1-t_2}{t_1}\approx 8.2\cdot 10^{-1}.
$$

The slight difference in the edge importance vectors computed using  \eqref{Lf} or 
\eqref{Lfapprox}  gives rise to different orderings of the edges corresponding to the 
$10$ smallest entries (which in fact differ about ${\cal O}(10^{-12})$). This is
displayed in Table \ref{T_4}. Removing the two edges in bold face in the second  
column of Table  \ref{T_4} results in  the (disconnected) network ${\mathcal G}_1$ 
(see Figure \ref{FIG1}(a)), while removing the two edges in bold face in the third 
column of Table  \ref{T_4} yields the disconnected network ${\widetilde{\mathcal G}}_1$
(see Figure \ref{FIG1}(b)). Nevertheless, we have the same result for both the graphs
${\mathcal G}_1$ and ${\widetilde{\mathcal G}}_1$:
\[
T_{{\mathcal G}_1}= 1.2550\cdot 10^{4};\;\;\;\;
\frac{ T_{\mathcal G}-T_{{\mathcal G}_1} } {T_{\mathcal G}} =
1.0171\cdot 10^{-3} ;\;\;\;\;\;\;\;\;
T_{{\widetilde{\mathcal G}}_1}= 1.2550\cdot 10^{4};\;\;\;\;
\frac{ T_{\mathcal G}-T_{{\widetilde{\mathcal G}}_1} } {T_{\mathcal G}} =
1.0171\cdot 10^{-3} .\;\;\;\;\;\;\;\;
\]

\begin{figure}[h!tb]
\centering
\begin{tabular}{cc}
\includegraphics[scale = 0.40]{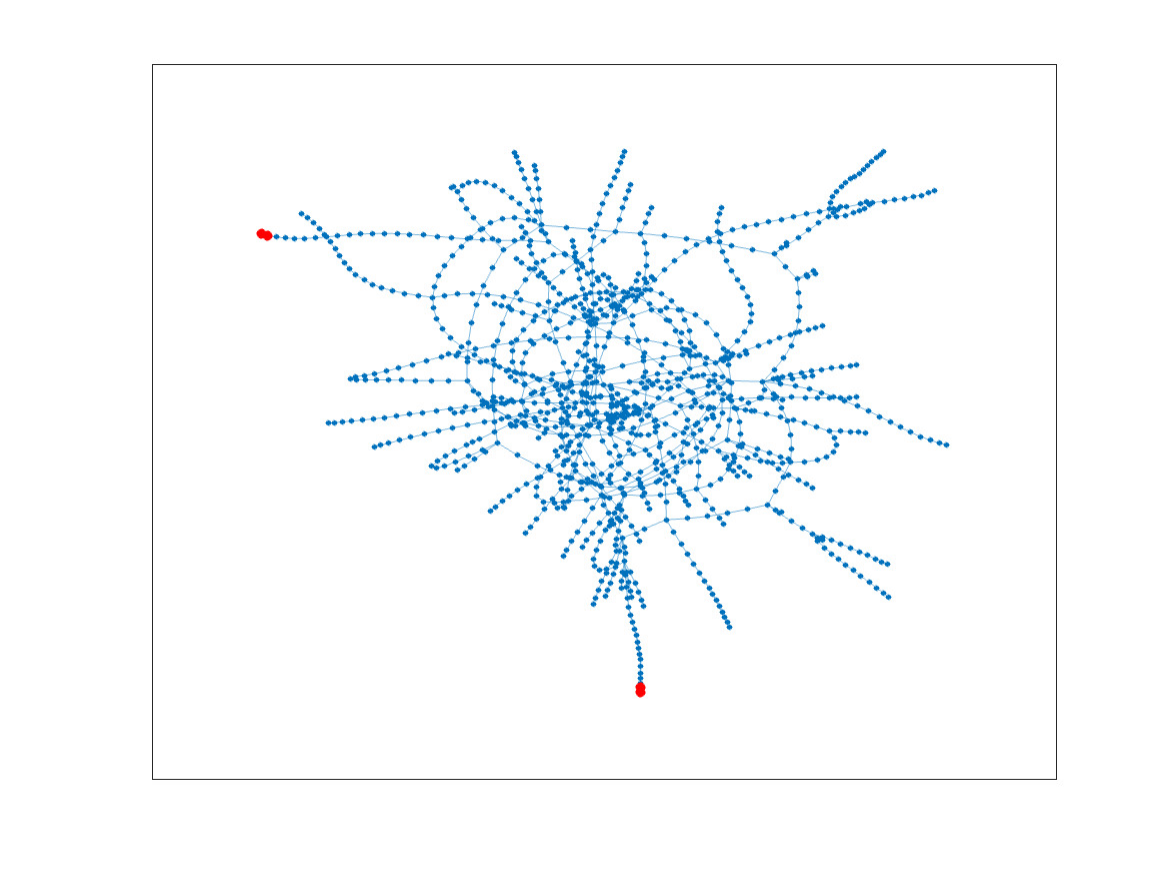} & 
\includegraphics[scale = 0.40]{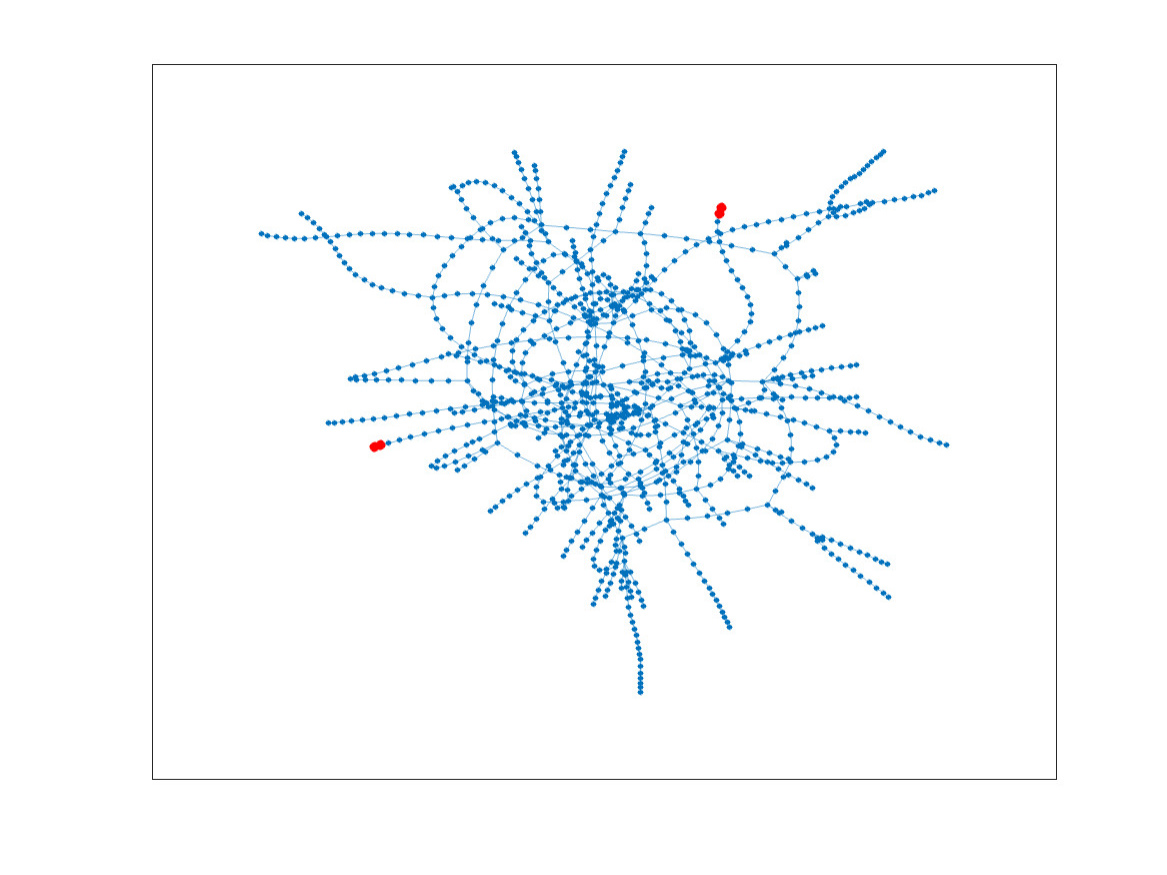} \\
 (a) & (b) 
\end{tabular}
\begin{tabular}{c}
\includegraphics[scale = 0.40]{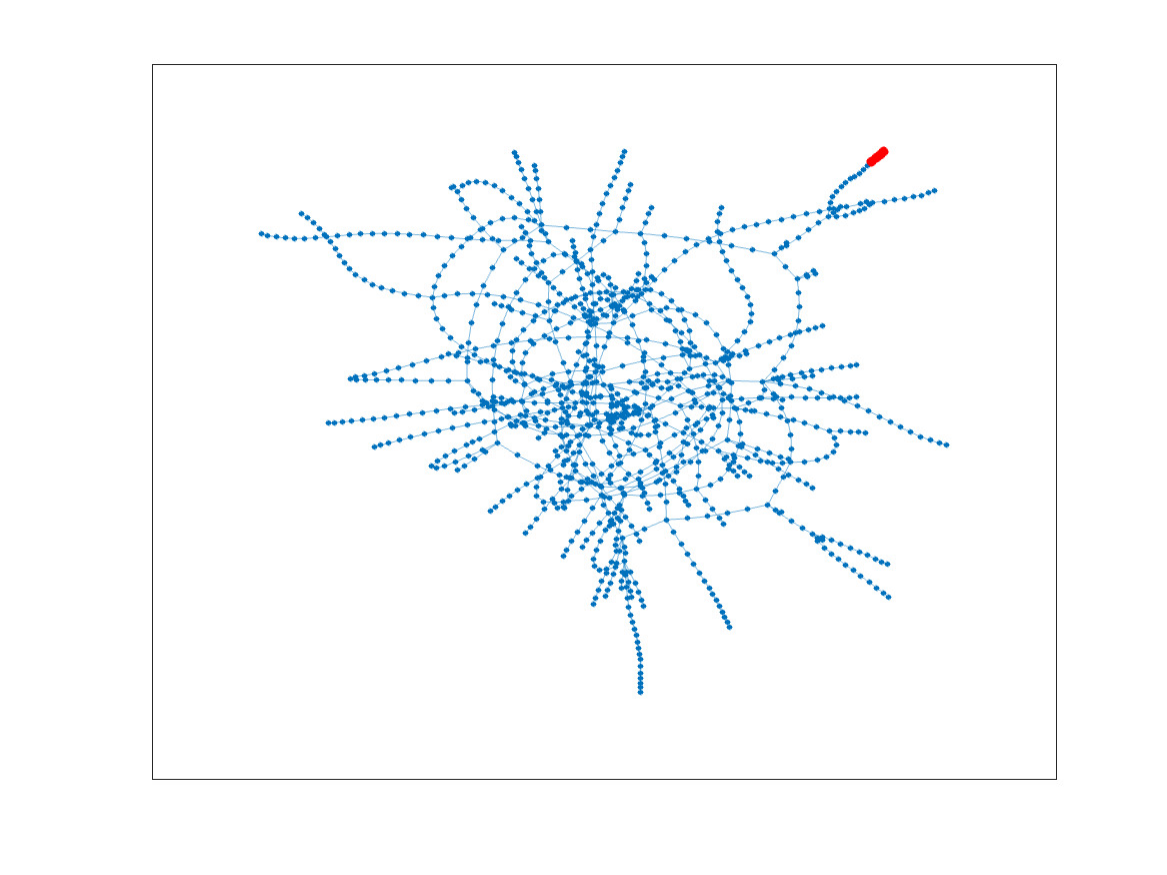} \\
 (c)
\end{tabular}
\caption{Example \ref{ex2}. The vertices (marked in red) that are connected by the edges to be  
removed in order to simplify the network according to ${\mathcal{E}}_{L_{f}}$ for (a), according to 
${\widetilde{\mathcal{E}}}_{L_{f}}$ for (b), and according to  ${\mathcal{E}}_{\rho}$ for (c).}
\label{FIG1}
\end{figure}
Conversely, the largest entries of the edge importance vectors
computed using  \eqref{Lf} or \eqref{Lfapprox} correspond to the same edges 
displayed in the first column of Table \ref{T_5}. 
Setting to zero the weights associated with the two highway segments  Duisburg - D\"usseldorf 
 and M\"unchen - Kirchheim (both in bold face in the first column of Table \ref{T_5}) results in the 
 network  ${\mathcal G}_2$, for which one has
 \[
T_{{\mathcal G}_2}= 1.2108\cdot 10^{4};\;\;\;\;
\frac{ T_{\mathcal G}-T_{{\mathcal G}_2} } {T_{\mathcal G}} =
3.6214\cdot 10^{-2},\;\;\;\;\;\;\;\;
\]
while increasing by one the weights associated with the same edges results in the network 
${\mathcal G}_3$, for which one has
\[
T_{{\mathcal G}_3}= 1.3359\cdot 10^{4};\;\;\;\;
\frac{ T_{{\mathcal G}_3} -T_{\mathcal G} } {T_{\mathcal G}} =
6.3330\cdot 10^{-2}.\;\;\;\;\;\;\;\;
\]

As for Perron communicability in {\em Autobahn}, one has $P_{\mathcal G}=2.2448\cdot 10^{3}$. 
Although the graph ${\mathcal G}$ is irreducible, some entries of the Perron vector $x$ are close 
to machine precision. In particular, the edges to remove in order to simplify the network, that are 
associated with the two smallest entries of  the Perron edge importance vector, are the 
highway segments Wildsdruff - Wildeck and  W\"ustenbrand - Wommen, which connect four 
vertices associated with such quasi-zero components of $x$ (see Figure \ref{FIG1}(c)).  This results in the 
(disconnected) network ${\mathcal G}_4$, for which we have
\[
P_{{\mathcal G}_4}=2.2448\cdot 10^{3};\;\;\;\;
\frac{ P_{{\mathcal G}}-P_{{\mathcal G_4}} } {P_{\mathcal G}} =
-4.8619\cdot 10^{-15};\;\;\;\;\;\;\;\;
T_{{\mathcal G}_4}= 1.2547\cdot 10^{4};\;\;\;\;
\frac{ T_{{\mathcal G}}-T_{{\mathcal G}_4} } {T_{\mathcal G}} =
1.2521\cdot 10^{-3}.
\]
 Hence, this simplification is slightly less satisfactory than the ones described by 
${\mathcal G}_1$ and
 ${\widetilde{\mathcal G}}_1$. Elimination of the edges associated with the two highway segments  
 Duisburg - D\"usseldorf and Essen - Duisburg (both in bold face in the second column of Table \ref{T_5}) 
 gives the network  ${\mathcal G}_5$, for which we have
\[
P_{{\mathcal G}_5}=1.2677\cdot 10^{3};\;\;\;\;
\frac{ P_{{\mathcal G}}-P_{{\mathcal G}_5} } {P_{\mathcal G}} =
4.3527\cdot 10^{-1};\;\;\;\;\;\;\;\;
T_{{\mathcal G}_5}= 1.2155\cdot 10^{4};\;\;\;\;
\frac{ T_{{\mathcal G}}-T_{{\mathcal G}_5} } {T_{\mathcal G}} =
3.2513\cdot 10^{-2}.
\]
Again the decrease is less than for ${\mathcal G}_2$. Conversely, increasing by one the weights 
associated with the same edges results in the network ${\mathcal G}_6$, for which one has
\[
P_{{\mathcal G}_6}=2.5248\cdot 10^{3};\;\;\;\;
\frac{ P_{{\mathcal G}_6}-P_{{\mathcal G}} } {P_{\mathcal G}} =
1.2472\cdot 10^{-2};\;\;\;\;\;\;\;\;
T_{{\mathcal G}_6}= 1.3480\cdot 10^{4};\;\;\;\;
\frac{ T_{{\mathcal G}_6}-T_{{\mathcal G}} } {T_{\mathcal G}} =
7.2982\cdot 10^{-2}.
\]
The increase in this case is greater than for ${\mathcal G}_3$. 

We do not address the issue of adding new highway segments, because the feasibility
of building new highway segments depends on many issues that are not included in our
model such as the length of the new segments and properties of the regions,
e.g., the presence of mountains, valleys and lakes, that the new segments would 
traverse.

\begin{table}[h!tb]
\centering
\begin{tabular}{c| l |l}
${\mathcal{E}}_{L_{f}}$ [or ${\widetilde{\mathcal{E}}}_{L_{f}}$]  & $e(v_i\leftrightarrow v_j)$ 
associated with ${\mathcal{E}}_{L_{f}}$ & $e(v_i\leftrightarrow v_j)$  associated with 
${\widetilde{\mathcal{E}}}_{L_{f}}$ \\
\hline
$6.5794\cdot 10^{-4}$  & \bf{Allershausen}  $\longleftrightarrow$ Allersberg   & 
\bf{W\"ustenbrand}  $\longleftrightarrow$ Wommen   \\
$6.5794\cdot 10^{-4}$  & \bf{B\"unde} $\longleftrightarrow$ Bissendorf  & 
\bf{Thiendorf}  $\longleftrightarrow$ Teupitz   \\
$6.5794\cdot 10^{-4}$  & Zarrentin $\longleftrightarrow$ Witzhave  &  
Zrbig  $\longleftrightarrow$ Wiedemar  \\
$6.5794\cdot 10^{-4}$  & Wunsiedel  $\longleftrightarrow$ Wolnzach  &  
Zarrentin  $\longleftrightarrow$  Witzhave \\ 
$6.5794\cdot 10^{-4}$  & Thiendorf $\longleftrightarrow$ Teupitz  & 
Aitrach $\longleftrightarrow$  Aichstetten \\ 
$6.5794\cdot 10^{-4}$  & W\"ustenbrand $\longleftrightarrow$ Wommen  &  
Wesuwe  $\longleftrightarrow$ Weener   \\
$6.5794\cdot 10^{-4}$  & Alsfeld $\longleftrightarrow$ Achern   & 
Wunsiedel $\longleftrightarrow$ 
Wolnzach    \\
$6.5794\cdot 10^{-4}$  & Wesuwe $\longleftrightarrow$ Weener  & 
Zwingenberg  $\longleftrightarrow$ Zeppelinheim   \\      
$6.5794\cdot 10^{-4}$  & Zwingenberg $\longleftrightarrow$ Zeppelinheim    &  
B\"unde $\longleftrightarrow$  Bissendorf \\ 
$6.5794\cdot 10^{-4}$  & Zrbig  $\longleftrightarrow$ Wiedemar   & 
Alsfeld  $\longleftrightarrow$ Achern\\ 
\end{tabular}
\caption{Example \ref{ex2}. The first $10$ highway segments that could be removed in order to reduce 
the complexity of {\em Autobahn} without changing the network communicability significantly, 
according to the determination of the edge importance based on the gradient. The edges in the second and 
third columns are determined by \eqref{Lf} and \eqref{Lfapprox}, respectively.} 
\label{T_4}
\end{table}

\begin{table}[h!tb]
\centering
\begin{tabular}{l l | l l }
${\mathcal{E}}_{L_{f}}$     &   $e(v_i\leftrightarrow v_j)$  &  ${\mathcal{E}}_{\rho}$     &   $e(v_i\leftrightarrow v_j)$       \\
\hline
$2.1662\cdot 10^{-2}$  & \bf{Duisburg}  $\longleftrightarrow$ D\"usseldorf   &  $3.2881\cdot 10^{-1}$  & \bf{Duisburg}  $\longleftrightarrow$ D\"usseldorf    \\
$1.9101\cdot 10^{-2}$  &\bf{M\"unchen} $\longleftrightarrow$ Kirchheim  &$2.9148\cdot 10^{-1}$  &\bf{Essen} $\longleftrightarrow$ Duisburg \\
$1.8792\cdot 10^{-2}$  &Essen $\longleftrightarrow$ Duisburg  &$2.2093\cdot 10^{-1}$  &Duisburg $\longleftrightarrow$ Dinslaken   \\
$1.8106\cdot 10^{-2}$  &Duisburg  $\longleftrightarrow$ Dortmund  &$2.1899\cdot 10^{-1}$  &Duisburg  $\longleftrightarrow$ Dortmund    \\
$1.7092\cdot 10^{-2}$  &Krefeld $\longleftrightarrow$ Duisburg  &$1.8552\cdot 10^{-1}$  &D\"usseldorf $\longleftrightarrow$ Dinslaken\\
$1.6862\cdot 10^{-2}$  &Hamburg $\longleftrightarrow$ Hagen  &$1.8126\cdot 10^{-1}$  &Krefeld $\longleftrightarrow$ Duisburg  \\
$1.4823\cdot 10^{-2}$  &Duisburg $\longleftrightarrow$ Dinslaken   &$1.5044\cdot 10^{-1}$  &Gelsenkirchen$\longleftrightarrow$ Essen    \\
$1.4348\cdot 10^{-2}$  &Gelsenkirchen $\longleftrightarrow$ Essen  &$1.4553\cdot 10^{-1}$  &Flughafen $\longleftrightarrow$ Duisburg  \\      
$1.4105\cdot 10^{-2}$  &Flughafen $\longleftrightarrow$ Duisburg    &$1.3617\cdot 10^{-1}$  &Elmpt $\longleftrightarrow$ D\"usseldorf   \\
$1.3341\cdot 10^{-2}$  &Hagen  $\longleftrightarrow$ Gro\'a   &$1.2071\cdot 10^{-1}$  &Essen  $\longleftrightarrow$ Elmpt\\
\end{tabular}
\caption{Example \ref{ex2}. The first $10$ highway segments that should be widened/narrowed
in order to increase/decrease the network communicability in {\em Autobahn} according to the 
edge importance vector ${\mathcal{E}}_{L_{f}}$ (in the first column) and according to 
the Perron edge importance  vector ${\mathcal{E}}_{\rho}$ (in the second column), respectively.}
\label{T_5}
\end{table}
\end{example}

\begin{example}\label{ex3}
Consider the directed weighted graph $\mathcal G$ that represents the  network {\it C.elegans} 
available at  \cite{celegans306_usroads48}, i.e.,  the metabolic network of the Caenorhabditis 
elegans worm. The network contains $n=306$ vertices that represent neurons and $m=2345$ edges. 
Two neurons are connected if at least one synapse exists between them and the 
associated edge-weight is the number of  synapses. The network is disconnected.

The total communicability  \eqref{TC} is $T_{\mathcal G}=3.3401\cdot 10^{6}$.
The gradient $\nabla T_{\mathcal G}$ \eqref{grad} has been computed by evaluating 
the matrix $L_f(\A^T,uv^T)$ in \eqref{Lfexpan} using \eqref{Lf}. 
The total transmission is $\|\nabla T_{\mathcal G}\|_2=7.9032\cdot 10^{6}$. The 
gradient $\nabla T_{\mathcal G}$ has been approximated by evaluating $L_f(\A^T,uv^T)$ 
using \eqref{Lfapprox} with $h=\frac{2}{n}\cdot10^{-4}=6.5359\cdot 10^{-7}$ to obtain 
$\widetilde{\nabla T}_{\mathcal G}$.  The resulting total transmission is 
$\|\widetilde{\nabla T}_{\mathcal G}\|_2=7.9032\cdot 10^{6}$, having
\[
\frac{\|\nabla T_{\mathcal G}-\widetilde{\nabla T}_{\mathcal G} \|_2}
{{\|\nabla T_{\mathcal G}}\|_2}=4.4914\cdot 10^{-9}.
\]

As for both the edge importance vector  ${\mathcal{E}}_{L_{f}}$ and the virtual edge importance 
vector $\widehat{\mathcal{E}}_{L_{f}}$, one obtains the same results, displayed in Table \ref{T_7},
regardless of whether  $\nabla T_{\mathcal G}$ or $\widetilde{\nabla T}_{\mathcal G}$ is used.
We  remark that evaluating $L_f(\A^T,ee^T)$ by \eqref{Lf}
required the average time $t_1\approx 1.3\cdot 10^{-1}$ in $10^4$ tests, while evaluating $L_f(\A^T,ee^T)$ using \eqref{Lfapprox} 
required the average time $t_2\approx 1.7\cdot 10^{-2}$ in the same tests. The relative 
average time saving is
$$
\frac{t_1-t_2}{t_1}\approx 8.6\cdot 10^{-1}.
$$

Removing the two edges in bold face in the second  column of Table  \ref{T_7}, associated with
the smallest entries of the edge importance, results in the network ${\mathcal G}_1$ 
for which
\[
T_{{\mathcal G}_1}= 3.3401\cdot 10^{6};\;\;\;\;
\frac{ T_{\mathcal G}-T_{{\mathcal G}_1} } {T_{\mathcal G}} =
5.9879\cdot 10^{-7}.
\]
Setting to zero the weights associated with the two edges in bold face in the fourth column of 
Table \ref{T_7} results in the network  ${\mathcal G}_2$, for which one has
 \[
T_{{\mathcal G}_2}= 2.9282\cdot 10^{6};\;\;\;\;
\frac{ T_{\mathcal G}-T_{{\mathcal G}_2} } {T_{\mathcal G}} =
1.2332\cdot 10^{-1},\;\;\;\;\;\;\;\;
\]
while increasing by one the weights associated with the same edges results in the network 
${\mathcal G}_3$ with 
\[
T_{{\mathcal G}_3}= 3.8047\cdot 10^{6};\;\;\;\;
\frac{ T_{{\mathcal G}_3} -T_{\mathcal G} } {T_{\mathcal G}} =
1.3909\cdot 10^{-1}.\;\;\;\;\;\;\;\;
\]
Finally, setting to one the (vanishing) entries of $A$ associated with the two virtual edges
in bold face in the sixth column of Table \ref{T_7} yields the network ${\mathcal G}_4$, 
for which 
\[
T_{{\mathcal G}_4}= 7.3327\cdot 10^{6};\;\;\;\;
\frac{ T_{{\mathcal G}_4} -T_{\mathcal G} } {T_{\mathcal G}} =
1.1954\cdot 10^{0}.\;\;\;\;\;\;\;\;
\]

Turning to network modifications based on Perron root sensitivity in {\em C.elegans}, one has 
$P_{\mathcal G}=9.1975\cdot 10^{5}$. 
The graph ${\mathcal G}$ is reducible, some entries of the Perron vector $x$ are vanishing. 
In particular, removing the edges in bold face in the second column of Table \ref{T_8}, which are 
associated with  two vanishing entries of  the Perron edge importance vector, results in the network 
${\mathcal G}_5$, for which one has
\[
P_{{\mathcal G}_5}=9.1931\cdot 10^{5};\;\;\;\;
\frac{ P_{{\mathcal G}}-P_{{\mathcal G_5}} } {P_{\mathcal G}} =
4.7732\cdot 10^{-4};\;\;\;\;\;\;\;\;
T_{{\mathcal G}_5}= 3.3381\cdot 10^{6};\;\;\;\;
\frac{ T_{{\mathcal G}}-T_{{\mathcal G}_5} } {T_{\mathcal G}} =
6.0666\cdot 10^{-4}.
\]
Hence, this simplification is less satisfactory than ${\mathcal G}_1$. 
Elimination of the edges  in bold face in the fourth column of Table \ref{T_8} gives 
the network ${\mathcal G}_6$, for which one has
\[
P_{{\mathcal G}_6}=7.4856\cdot 10^{5};\;\;\;\;
\frac{ P_{{\mathcal G}}-P_{{\mathcal G}_6} } {P_{\mathcal G}} =
1.8613\cdot 10^{-1};\;\;\;\;\;\;\;\;
T_{{\mathcal G}_6}= 2.9298\cdot 10^{6};\;\;\;\;
\frac{ T_{{\mathcal G}}-T_{{\mathcal G}_6} } {T_{\mathcal G}} =
1.2284\cdot 10^{-1}.
\]
The decrease is less than that for ${\mathcal G}_2$. Conversely, increasing by one the weights 
associated with the same edges results in the network ${\mathcal G}_7$, for which one has
\[
P_{{\mathcal G}_7}=1.1140\cdot 10^{6};\;\;\;\;
\frac{ P_{{\mathcal G}_7}-P_{{\mathcal G}} } {P_{\mathcal G}} =
2.1122\cdot 10^{-1};\;\;\;\;\;\;\;\;
T_{{\mathcal G}_7}= 3.7933\cdot 10^{6};\;\;\;\;
\frac{ T_{{\mathcal G}_7}-T_{{\mathcal G}} } {T_{\mathcal G}} =
1.3569\cdot 10^{-1}.
\]
The increase is less than that for ${\mathcal G}_3$. Finally, setting to one the 
(vanishing) entries of $A$ associated with the two virtual edges in bold face in the sixth
column of Table \ref{T_8} results in the network ${\mathcal G}_8$, having
\[
P_{{\mathcal G}_8}=3.5870\cdot 10^{6};\;\;\;\;
\frac{ P_{{\mathcal G}_8}-P_{{\mathcal G}} } {P_{\mathcal G}} =
2.8999\cdot 10^{0};\;\;\;\;\;\;\;\;
T_{{\mathcal G}_8}= 7.3364\cdot 10^{6};\;\;\;\;
\frac{ T_{{\mathcal G}_8}-T_{{\mathcal G}} } {T_{\mathcal G}} =
1.1965\cdot 10^{0}.
\]
Therefore the result is more satisfactory than that obtained for the network  ${\mathcal G}_4$.
Notice that $\|\W|_{{\mathcal{A}}} \|_F=1.2817\cdot 10^{-1}$.

\begin{table}[h!tb]
\centering
\begin{tabular}{l l | l l | l l}
${\mathcal{E}}_{L_{f}}$       & $e(v_i\to v_j)$ & ${\mathcal{E}}_{L_{f}}$& $e(v_i\to  v_j)$ 
& ${\widehat{\mathcal{E}}}_{L_{f}}$  & $e(v_i\to  v_j)$ \\
\hline
$1.2653\cdot 10^{-7}$  &\boldmath{$v_{53} \to v_{303}$}  & $2.8888\cdot 10^{-2}$  & 
 \boldmath{$v_{71} \to v_{217}$}  &$1.1107\cdot 10^{-1}$  &  \boldmath{$v_{305}$}  $\to $ 
 \boldmath{$v_{149}$}  \\
$1.2653\cdot 10^{-7}$   &\boldmath{$v_{151} \to v_{305}$} & $2.6446\cdot 10^{-2}$  & 
\boldmath{$v_{72} \to  v_{216}$} & $1.0213\cdot 10^{-1}$  & \boldmath{$v_{305} \to 
v_{219}$}  \\
$1.2653\cdot 10^{-7}$  & $v_{191}$ $\to $ $v_{305}$  & $2.5673\cdot 10^{-2}$  & 
$v_{73} $ $\to $ $v_{178}$ & $9.6897\cdot 10^{-2}$  & $v_{305}$  $\to $ $v_{218}$  \\
$1.2653\cdot 10^{-7}$  & $v_{243}$ $\to $ $v_{305}$  & $2.2681\cdot 10^{-2}$  & 
$v_{72}$  $\to $ $v_{144}$  &$8.9795\cdot 10^{-2}$  & $v_{305} $ $\to $ $v_{216}$ \\
$1.2653\cdot 10^{-7}$  & $v_{259}$ $\to $ $v_{305}$  & $2.1866\cdot 10^{-2}$  & 
$v_{76} $ $\to $ $v_{217}$  &$8.9722\cdot 10^{-2}$  & $v_{305}$ $\to $ $v_{217}$  \\
$1.2653\cdot 10^{-7}$  & $v_{267}$ $\to $ $v_{305}$  & $2.1271\cdot 10^{-2}$  & 
$v_{78} $ $\to $ $v_{217}$  &$8.9222\cdot 10^{-2}$  & $v_{305}$  $\to $ $v_{178} $ \\
$1.2653\cdot 10^{-7}$  & $v_{291} $ $\to $ $v_{305} $  & $2.0888\cdot 10^{-2}$  & 
$v_{75}$  $\to $ $v_{216}$ &$8.5450\cdot 10^{-2}$  & $v_{305}$ $\to $ $v_{174} $ \\
$1.2653\cdot 10^{-7}$  & $v_{292}$ $\to $ $v_{305}$  & $2.0012\cdot 10^{-2}$  & 
$v_{77} $ $\to $ $v_{216}$  &$8.2518\cdot 10^{-2}$  & $v_{305}$  $\to $ $v_{81} $ \\
$1.2653\cdot 10^{-7}$  & $v_{293} $ $\to $ $v_{305} $   & $1.9903\cdot 10^{-2}$  & 
$v_{71}$  $\to $ $v_{47} $ &$7.9914\cdot 10^{-2}$  & $v_{305}$ $\to $ $v_{82} $ \\
$1.2653\cdot 10^{-7}$  & $v_{294} $ $\to $ $v_{305} $  & $1.9752\cdot 10^{-2}$  &
$ v_{71} $ $\to $ $v_{72}$  &$7.7984\cdot 10^{-2}$  &$v_{305} $ $\to $ $v_{198} $ \\
\end{tabular}
\caption{Example \ref{ex3}. The smallest entries of the edge importance vector and the 
relevant  junctions that could be removed without changing the network communicability 
in {\it C.elegans} significantly (displayed in the first two columns). The largest entries of 
the edge importance vector and the relevant  junctions to increase/decrease in order to 
increase/decrease the network communicability (in the third and fourth columns). The 
largest entries of the virtual edge importance vector and the relevant junctions to add in
order to increase the communicability in {\it C.elegans} (in the fifth and sixth columns).} 
\label{T_7}
\end{table}

\begin{table}[h!tb]
\centering
\begin{tabular}{l l | l l | l l}
${\mathcal{E}}_{\rho}$       & $e(v_i\to v_j)$ & ${\mathcal{E}}_{\rho}$& $e(v_i\to  v_j)$ 
& ${\widehat{\mathcal{E}}}_{\rho}$  & $e(v_i\to  v_j)$ \\
\hline
$0$  &\boldmath{$ v_{53} \to v_{1} $} & $2.7208\cdot 10^{-2}$  &\boldmath{$v_{73} \to v_{178}$}  &
$1.3484\cdot 10^{-1}$  &\boldmath{$v_{305} \to v_{149}$}  \\
$0$  &\boldmath{$v_{11} \to v_{5}$}& $2.4126\cdot 10^{-2}$  & \boldmath{$v_{71} \to v_{217}$}  &
$1.3000\cdot 10^{-1}$  &\boldmath{$v_{305} \to v_{219}$}  \\
$0$  & $v_{12}$ $\to $ $v_{6}$  & $2.2348\cdot 10^{-2}$  & $v_{72} $ $\to $ $ v_{216}$ &
$1.2205\cdot 10^{-1}$  & $v_{305} $ $\to $ $ v_{218} $ \\
$0$  & $v_{11} $ $\to $ $ v_{19}$  & $2.0573\cdot 10^{-2}$  & $v_{72}  $ $\to $ $ v_{144} $ &
$1.1067\cdot 10^{-1}$  &$v_{305} $ $\to $ $ v_{178}$\\
$0$  & $v_{11} $ $\to $ $v_{23}$ & $1.8951\cdot 10^{-2}$  & $v_{76}  $ $\to $ $ v_{217}  $&
$1.0811\cdot 10^{-1}$  & $v_{305}$ $\to $ $ v_{174}$ \\
$0$  & $v_{12} $ $\to $ $ v_{24}$  & $1.8489\cdot 10^{-2}$  & $v_{71}  $ $\to $ $ v_{47} $&
$1.0128\cdot 10^{-1}$  & $v_{305}$  $\to $ $ v_{81}$  \\
$0$  & $v_{12} $ $\to $ $ v_{25 }$  & $1.8419\cdot 10^{-2}$  & $v_{75}  $ $\to $ $ v_{216} $&
$9.7793\cdot 10^{-2}$  & $v_{305}$ $\to  $ $ v_{82} $\\
$0$  & $v_{8}$ $\to $ $ v_{26} $ & $1.7529\cdot 10^{-2}$  & $v_{74}  $ $\to $ $ v_{177}  $&
$9.5414\cdot 10^{-2}$  & $v_{305} $ $\to $ $ v_{157} $ \\
$0$  & $v_{11} $ $\to $ $ v_{26} $  & $1.7392\cdot 10^{-2}$  & $v_{216}$  $\to $ $ v_{81}  $&
$8.7441\cdot 10^{-2}$  & $v_{305} $ $\to $ $ v_{216}  $\\
$0$  & $v_{12} $ $\to $ $ v_{26} $ & $1.7269\cdot 10^{-2}$  & $v_{78}$  $\to $ $ v_{217} $&
$8.6713\cdot 10^{-2}$  &$v_{305} $ $\to $ $ v_{217} $ \\
\end{tabular}
\caption{Example \ref{ex3}. The smallest entries of the Perron edge importance vector  
and the relevant  junctions that could be removed without changing the  communicability 
in {\it C.elegans} significantly (displayed in the first two columns). The largest entries of 
the Perron edge importance vector and the relevant  junctions to increase/decrease in  
order to increase/decrease the network communicability in {\it C.elegans} (in the third and 
fourth columns). The largest entries of the virtual Perron edge importance vector and the 
relevant junctions to add in order to increase the network communicability (displayed in the 
fifth and sixth columns).} 
\label{T_8}
\end{table}
\end{example}

\subsection{Large networks}\label{subsec42}
\begin{example}\label{ex4}
Consider the unweighted undirected  graph $\mathcal G$ that represents the continental US 
road network {\it Usroads-48}. The graph $\mathcal G$, which is available at 
\cite{celegans306_usroads48},  has $n=126146$ vertices, which represent intersections and 
road endpoints. The $m=161950$ edges represent roads that connect the intersections and 
endpoints. 
We analyze the network {\it Usroads-48} with the tools discussed in Subsections  \ref{subsec2.2} 
and \ref{subsec32} for large networks. 

We would like to determine approximations of the smallest and largest elements of 
$L_{f}(A, ee^T)|_{\mathcal L}$  
that correspond to  edges that should be removed to simplify the network or whose edge-weight 
should be modified to increase or decrease the total communicability. Moreover, we would like
to determine approximations of the smallest and largest elements of 
$L_{f}(A, ee^T)|_{\widehat{\mathcal{L}}}$  
that correspond to  edges that should be added  to  increase total communicability. 
We first carry out $5$ steps of the symmetric Lanczos process, computing the matrices 
$\U_5$ and $\T_5$, and make use of the latter to construct both the matrix $\X_5$ in 
\eqref{Lfapprox2L} by using \eqref{Lfapproxsmall}, and the matrix 
$${\tilde \X}_5=\frac{f(\T_5 + hne_1e_1^T)-f(\T_5 - hne_1e_1^T)}{2h},$$
with $h=\frac{2}{n}\cdot 10^{-4}$. Then, proceeding as described in Subsection 
\ref{subsec2.2}, we found that both the matrices  $\X_5$ and ${\tilde \X}_5$ 
determined the same edges. Regarding the timings (in seconds) of the two procedures,
we remark that, having available the symmetric and tridiagonal matrix $\T_{5}$ determined 
by $\ell=5$ 
steps of  the Lanczos process, evaluating the matrix $\X_5$ in \eqref{Lfapproxsmall}
required the average time $t_1\approx2.9\cdot 10^{-5}$ over $10^4$ tests, while evaluating the matrix ${\tilde \X}_5$ 
required the average time $t_2\approx2.7\cdot 10^{-5}$ in the same tests, with a relative 
average time saving
$$
\frac{t_1-t_2}{t_1}\approx7.0\cdot 10^{-2}.
$$

The smallest element of the computed approximation of
$L_{f}(A, ee^T)|_{{\mathcal{L}}}$ is $9.7462\cdot 10^{-7}$ and is associated with the
edge $e(v_{123259}\leftrightarrow  v_{123258})$, while the largest element is 
$1.5008\cdot 10^{1}$ and is associated with the edge 
$e(v_{19694}\leftrightarrow  v_{19186})$; the smallest element of the  computed 
approximation of $L_{f}(A, ee^T)|_{\widehat{\mathcal{L}}}$ is $2.2047\cdot 10^{-9}$; 
it is associated with edge the $e(v_{25416}\leftrightarrow  v_{11651})$. The 
largest element is $1.9380\cdot 10^{1}$ and is associated with the edge 
$e(v_{58080}\leftrightarrow  v_{1})$.

Turning to the the structured Perron communicability, one  has  $P_{\mathcal G}^{{\mathcal{A}}} 
=1.9138\cdot10^{2}$. The smallest entries of the Perron edge importance vector 
${\mathcal{E}}_{\rho}$ and the relevant edges are displayed in Table \ref{T_9} in the first and 
second columns, while the largest entries of  ${\mathcal{E}}_{\rho}$  and the relevant edges 
are shown in the third and fourth columns. 
In order to increase the network communicability one should add edge 
$e(v_{44182}\leftrightarrow  v_{44035})$, which is associated with the largest entry of the 
{$\widehat{\mathcal{L}}$}-analogue of the Wilkinson perturbation.

\begin{table}[h!tb]
\centering
\begin{tabular}{l l | l l }
${\mathcal{E}}_{\rho}$       & $e(v_i\leftrightarrow v_j)$ & ${\mathcal{E}}_{\rho}$& $e(v_i\leftrightarrow  v_j)$  \\
\hline
$1.2197\cdot10^{-40}$  &\boldmath{$ v_{105751} \leftrightarrow v_{105743} $} & $2.2346\cdot 10^{-1}$
  &\boldmath{$v_{44182} \leftrightarrow v_{44067}$}   \\
$3.3336\cdot10^{-40}$  &\boldmath{$v_{42664} \leftrightarrow v_{42479}$}& $1.8605\cdot 10^{-1}$  
& \boldmath{$v_{44182} \leftrightarrow v_{44154}$}   \\
$5.0488\cdot10^{-40}$  & $v_{114032}$ $\leftrightarrow $ $v_{42664}$  & $1.8090\cdot 10^{-1}$  
& $v_{44182} $ $\leftrightarrow $ $ v_{44087}$\\
$7.3043\cdot10^{-40}$  & $v_{44150} $ $\leftrightarrow $ $ v_{44015}$  & $1.5882\cdot 10^{-1}$  
& $v_{44154}  $ $\leftrightarrow $ $ v_{44067} $ \\
$8.8379\cdot10^{-40}$  & $v_{68387} $ $\leftrightarrow $ $v_{68213}$ & $1.5443\cdot 10^{-1}$  
& $v_{44087}  $ $\leftrightarrow $ $ v_{44067}  $\\
$1.0103\cdot10^{-39}$  & $v_{39830} $ $\leftrightarrow $ $ v_{39787}$  & $1.5077\cdot 10^{-1}$  
& $v_{44323}  $ $\leftrightarrow $ $ v_{44182} $  \\
$1.0695\cdot10^{-39}$  & $v_{29088} $ $\leftrightarrow $ $ v_{29056 }$  & $1.4155\cdot 10^{-1}$  
& $v_{44255}  $ $\leftrightarrow $ $ v_{44182} $\\
$1.2043\cdot10^{-39}$  & $v_{90123}$ $\leftrightarrow $ $ v_{89379} $ & $1.4099\cdot 10^{-1}$  
& $v_{44356}  $ $\leftrightarrow $ $ v_{44182}  $\\
$1.2478\cdot10^{-39}$  & $v_{78533} $ $\leftrightarrow $ $ v_{78388} $  & $1.1663\cdot 10^{-1}$  
& $v_{44067}$  $\leftrightarrow $ $ v_{44035}  $\\
$1.2821\cdot10^{-39}$  & $v_{35630} $ $\leftrightarrow $ $ v_{35115} $ & $9.2663\cdot 10^{-2}$  
& $v_{44067}$  $\leftrightarrow $ $ v_{44019} $ \\
\end{tabular}
\caption{Example \ref{ex4}. The smallest entries of the Perron edge importance vector  
and the edges associated with the roads that could be removed in order to reduce 
the complexity of the network without changing the communicability in {\it Usroads-48} significantly 
(displayed in the first two columns). The largest entries of the Perron edge importance 
vector and the edges associated with the first ten roads that should be widened/narrowed
 in order to increase/decrease the network communicability the most (in the third and fourth 
columns).} 
\label{T_9}
\end{table}

Removing the two edges in bold face in the second column of Table  \ref{T_9}, associated 
with the smallest entries of the edge importance, gives the network ${\mathcal G}_1$, 
for which  one has
\[
P_{{\mathcal G}_1}^{\cal A} = 1.9138\cdot 10^{2};\;\;\;\;
\frac{ P_{\mathcal G}^{\cal A} -P_{{\mathcal G}_1}^{\cal A} }
{P_{\mathcal G} ^{\cal A}} = 1.4851\cdot 10^{-16}.
\]
Setting to zero the weights associated with the two edges in bold face in the fourth column of 
Table \ref{T_9} results in the network  ${\mathcal G}_2$ with
 \[
P_{{\mathcal G}_2}^{\cal A} = 1.7487\cdot 10^{2};\;\;\;\;
\frac{ P_{\mathcal G}^{\cal A} -P_{{\mathcal G}_2}^{\cal A}} 
{P_{\mathcal G}^{\cal A}} = 8.6228\cdot 10^{-2},
\]
while increasing the weights associated with the same edges by one results in the network 
${\mathcal G}_3$, for which one has
\[
P_{{\mathcal G}_3}^{\cal A} = 3.6445\cdot 10^{2};\;\;\;\;
\frac{ P^{\cal A}_{{\mathcal G}_3} -P_{\mathcal G}^{\cal A}} 
{P_{\mathcal G} ^{\cal A}} = 9.0435\cdot 10^{-1}.
\]
Finally, setting to one the (vanishing) entry $w_{44182,44035}$ of the adjacency matrix $\A$ 
associated with the virtual edge $e(v_{44182}\leftrightarrow  v_{44035})$ returns the network 
${\mathcal G}_4$, with
\[
P^{\cal A}_{{\mathcal G}_4} = 2.4540\cdot 10^{2};\;\;\;\;
\frac{ P^{\cal A}_{{\mathcal G}_4}-P^{\cal A}_{\mathcal G}} 
{P^{\cal A}_{\mathcal G} } = 2.8228\cdot 10^{-1}.
\]
Notice that one has  $\|\W|_{{\mathcal{A}}} \|_F=\sqrt{2}\|\W|_{{\mathcal{L}}}\|_F=
2.9318\cdot 10^{-1}$.
\end{example}

\section{Conclusion and comments on related work}\label{sec6} 
The identification of important and unimportant edges is a fundamental problem in network 
analysis.  Several techniques for this purpose have been described in the literature; see,
e.g., \cite{DLCJR,DLCMR2,DLCMR3,EHNR,NR,NR2}. In \cite{DLCJNR} the authors propose a 
method that uses the gradient of the total communicability, and Schweitzer \cite{Sc} 
recently described how the computational effort required by this method can be reduced. 
Section \ref{sec2} of this paper reviews this method and discusses computational aspects 
when this technique is applied to small and medium-sized networks, as well as to 
large-scale networks. In particular, further ways to speed up the computations when the 
method is applied to large-scale networks are described. 

Another approach to identify important and unimportant edges is to determine edge weights
whose modification yields a relatively large change in the Perron root of the 
adjacency matrix. This is described in \cite{EHNR,NR}. The computations required are quite
straightforward and the method discussed in the latter reference is easy to implement also
for large-scale problems. We therefore are  interested in whether modifications of the 
weights of the edges identified by this technique give a relatively large change in the 
total communicability. Section \ref{sec3} reviews the method described in \cite{NR} and 
extends it to include edge removal. Computed examples reported in Section \ref{sec4} show 
that, indeed, modifications of edge weights identified by the technique discussed in 
\cite{NR} yield relatively large changes in the total communicability.

\end{document}